\documentclass[submission,copyright,creativecommons]{eptcs}
\providecommand{\event}{Linearity \& TLLA 2020} % Name of the event you are submitting to
\usepackage{breakurl}             % Not needed if you use pdflatex only.
\usepackage{underscore}           % Only needed if you use pdflatex.

\usepackage{amscd}
\usepackage{amsmath}
\usepackage{amssymb}
\usepackage{amsthm}
\usepackage{proof}
\usepackage{bussproofs}
\usepackage[all]{xy}
\usepackage{cmll}
\usepackage{shuffle}
\usepackage{graphicx}

\newtheorem{definition}{Definition}[section]
\newtheorem{theorem}[definition]{Theorem}%[subsection]
\newtheorem{lemma}[definition]{Lemma}%[subsection]
\newtheorem{corollary}[definition]{Corollary}%[subsection]
\newtheorem{proposition}[definition]{Proposition}%[subsection]
\newtheorem{example}[definition]{Example}%[subsection]
\title{Harmony in the Light of Computational Ludics\thanks{The authors thank the anonymous reviewers for their valuable comments on the early version of this paper. Alberto Naibo's work is partially supported by the ANR project \textit{PROGRAMme} (ANR-17-CE38-0003-01), by the ANR-DFG project \textit{FFIUM} (ANR-17-FRAL-0003) and by the ANR project \textit{GoA} (ANR-20-CE27-0004). Yuta Takahashi's work is supported by JSPS KAKENHI Grant Number JP21K12822. Part of the present work has been prepared during Yuta Takahashi's stay at the IHPST as a postdoctoral researcher, with the support of JSPS (Japan Society for the Promotion of Science) Overseas Research Fellowship (July 2019--June 2021).}}
\author{Alberto Naibo
\institute{IHPST (UMR 8590), Universit\'{e} Paris 1 Panth\'{e}on-Sorbonne, CNRS\\Paris, France}
\email{alberto.naibo@univ-paris1.fr}
\and
Yuta Takahashi
\institute{Ochanomizu University\\Tokyo, Japan}
\email{takahashi.yuta@is.ocha.ac.jp}
}
\def\titlerunning{Harmony in the Light of Computational Ludics}
\def\authorrunning{A. Naibo \& Y. Takahashi}

\newcommand{\Gam}{\Gamma}
\newcommand{\Del}{\Delta}

\newcommand{\uhr}{\upharpoonright}

\newcommand{\bd}[1]{\mathbf{#1}}
\newcommand{\fk}[1]{\mathfrak{#1}}

\newcommand{\otv}[1]{[\![ #1 ]\!]}
\newcommand{\val}[1]{\vert #1 \vert}
\newcommand{\rwt}{\longrightarrow}

\newcommand{\mall}{\mathbf{MALL}}
\newcommand{\mallp}{\mathbf{MALLP}}
\newcommand{\ssse}{\sqsubseteq}

\newcommand{\cA}{\mathcal{A}}

\newcommand{\cD}{\mathcal{D}}
\newcommand{\sC}{\mathsf{C}}

\newcommand{\fD}{\mathfrak{D}}
\newcommand{\fE}{\mathfrak{E}}

\newcommand{\cT}{\mathcal{T}}

\newcommand{\fv}[1]{\mathsf{fv}(#1)}
\newcommand{\np}[1]{\mathsf{np}(#1)}

\newcommand{\ov}[1]{\overline{#1}}
\newcommand{\view}[1]{\ulcorner #1 \urcorner}
\newcommand{\viewd}[1]{\mathbb{V}( #1 )}

\newcommand{\antiview}[1]{\llcorner #1 \lrcorner}
\newcommand{\dual}[1]{\widetilde{ #1 }}
\newcommand{\iseq}[2]{\langle #1 \leftarrow #2 \rangle}

\newcommand{\Cut}[2]{\mathfrak{Cut}_{ #1 | #2 }}

\begin{document}
\maketitle

\begin{abstract}
  Prawitz formulated the so-called inversion principle as one of the characteristic features of Gentzen's intuitionistic natural deduction. In the literature on proof-theoretic semantics, this principle is often coupled with another that is called the recovery principle. By adopting the Computational Ludics framework, we reformulate these principles into one and the same condition, which we call the harmony condition. We show that this reformulation allows us to reveal two intuitive ideas standing behind these principles: the idea of ``containment" present in the inversion principle, and the idea that the recovery principle is the ``converse" of the inversion principle. We also formulate two other conditions in the Computational Ludics framework, and we show that each of them is equivalent to the harmony condition.
\end{abstract}

%%%
%%%
%%% New Section
%%%
%%%
\section{Introduction}\label{intro}
This paper aims to study some of the characteristic features of the so-called \textit{proof-theoretic semantics} within the framework of Computational Ludics. Generally, the main objective of proof-theoretic semantics is to explain the meaning of linguistic expressions in terms of proof-conditions rather than truth-conditions, which are typical of referentialist semantics. In particular, by taking inspiration from the Brouwer-Heyting-Kolmogorov explanation of logical connectives, proof-theoretic semantics rests on the idea that we know the meaning of a compound sentence when we know what counts as a \textit{canonical proof} of it. And if proofs are formalised within the framework of natural deduction, then a canonical proof of a sentence $A$ is nothing but a closed derivation ending with an introduction rule of the main connective of $A$.\footnote{Following \cite[\S~1.3]{SchroederHeister18}, we say that a derivation in natural deduction is closed if it has no open assumption, otherwise we say that it is open.} 
The introduction rules play then a privileged role in fixing the meaning of a certain connective. It is in this sense that we should understand Gentzen's remark, according to which the introduction rules of a connective represent, as it were, the ``definitions'' of this connective, while the elimination rules of such a connective are nothing but the ``consequences'' of these definitions (see \cite[p.~80]{Gentzen69}). 

However, according to Prawitz, the words ``definition'' and ``consequence''  are used here only in a sort of metaphorical way \cite[p.~33, f.n.~1]{prawitz1965}. To assign a more precise sense to Gentzen's remark, Prawitz formulated the so-called \textit{inversion principle}. The idea behind this principle is that an elimination rule $E$ of a certain connective should essentially behave as the ``inverse" of the corresponding introduction rule(s), in the sense that by an application of $E$ one simply ``restores what had already been established if the major premise of the application was inferred by an application of an introduction rule'' \cite[p.~33]{prawitz1965}. 

We can explain how this idea is used by Prawitz to better specify Gentzen's remark by means of an example. Consider the introduction ($I$) and elimination ($E$) rules for the implication ($\to$):
\begin{center}
{\small
$\infer[ {\scriptstyle \to_I \; n} ]{A \to B}{
  \infer*[d]{B}{\stackrel{n}{A}}
}
\qquad\qquad
\infer[ {\scriptstyle \to_E } ]{B}{
  A \to B
  &
  A
}$
}
\end{center}
If we accept that the $\to_{I}$ rule defines the connective $\to$, in the sense that it determines the meaning of $\to$, then by stating $A \to B$ we should not be allowed to deduce anything \textit{more} than what we can already obtain from the sub-derivation $d$. Otherwise, the meaning of $A \to B$ would be more informative than what is stipulated by the $\to_{I}$ rule. It is in this sense that we should understand Prawitz's characterisation of the inversion principle in terms of ``containment'': the premise of an introduction rule of certain connective $\texttt{c}$ already contains all of the information that is required to obtain the conclusion of the corresponding elimination rule. The application of the elimination rule is thus dispensable when its major premise is the conclusion of an introduction rule. In the case of implication, this means that
\begin{center}
{\small 
$\infer[ {\scriptstyle \to_E} ]{B}{
  \infer[ {\scriptstyle \to_I \; n} ]{A \to B}{
    \infer*[d_1]{B}{\stackrel{n}{A}}
  }
  &
  \infer*[d_2]{A}{}
}
\qquad
\textnormal{can be transformed into}
\qquad
\infer*[d_1]{B}{\infer*[d_2]{[A]}{}}$
}
\end{center}
because the (open) derivation $d_{1}$ of the premise of the $\to_{I}$ rule, when combined with the derivation $d_{2}$ of the minor premise of the $\to_{E}$ rule, already ``contains" a derivation of the conclusion of $\to_{E}$ (see \cite[p.~33]{prawitz1965}). It is precisely in terms of this relation of ``containment" that one can make sense of Gentzen's remark that the elimination rules are nothing but the consequences of the introduction rules of a certain connective. The proof-transformation that we have just presented corresponds to what is usually called a (local) reduction step of a \textit{detour} (in this case, a $\to$-detour).\footnote{When Prawitz formulated the inversion principle in his monograph on natural deduction in 1965, he was unaware of Gentzen's unpublished works. In particular, Prawitz was not aware that in an unpublished version of his PhD thesis, which was only discovered in 2005 by Jan von Plato \cite{vonplato2008}, Gentzen had already defined the detours reduction steps for the rules of intuitionistic logic. It is certainly for this reason that he considered that Gentzen's remark on the relationship between introduction and elimination rules of natural deduction was stated in metaphorical terms (while we know today that Gentzen's remark was probably based on some technical results similar to those later obtained by Prawitz himself).} And if we work under the proofs-as-programs correspondence (i.e. the Curry-Howard correspondence), then a computational content can be assigned to this detour reduction step because it corresponds to a $\beta$-reduction step (preserving typing) in typed $\lambda$-calculus.  
 
Several subsequent works \cite{PfenningDavies2001,FrancezDyckhoff2012,SH2014} tried to improve Prawitz's analysis of Gentzen's remark by coupling the inversion principle with another principle, which we call here the \textit{recovery principle} by following \cite{SH2014}. As we have seen, the inversion principle corresponds to a \textit{no more} condition. However, if we want the $\to_{I}$ rule to \textit{completely} determine the meaning of the connective $\to$, then we have to impose an extra condition in addition to the one already imposed by the inversion principle. This extra condition consists in asking that by stating $A \to B$ we should not be allowed to deduce anything \textit{less} than what we can already obtain from the sub-derivation $d$. Otherwise, the meaning of $A \to B$ would be less informative than what is stipulated by the $\to_{I}$ rule. If we want to mimic Prawitz's account of the inversion principle and formulate the recovery principle in terms of ``containment", then we should do it in a fashion which looks like the ``converse" of the inversion principle itself: all of the information that is required to obtain the conclusion of an introduction rule of a certain connective $\texttt{c}$ is already contained in the derivation of the major premise of the corresponding elimination rule of $\texttt{c}$. In the case of implication, this means that 
\begin{center}
{\small
$\infer*[d]{A \to B}{}
\qquad
\textnormal{can be transformed into}
\qquad
\infer[ {\scriptstyle \to_I \; n} ]{A \to B}{
  \infer[ {\scriptstyle \to_E} ]{B}{
    \infer*[d]{A \to B}{}
    &
    \stackrel{n}{A}
  }
}$
}
\end{center}
The idea is that given a derivation $d$ of the major premiss of the $\to_{E}$ rule, one can extract from it everything that is required to apply the corresponding introduction rule $\to_{I}$.\footnote{The situation is more complicated for a disjunction. The particular format of its elimination rule means that to obtain $A \vee B$ from a derivation \AxiomC{$d$}\noLine\UnaryInfC{$A \vee B$}\DisplayProof, one first has  to apply a $\vee_{I}$ rule to obtain $A \vee B$ in each of the minor premisses of $\vee_{E}$, and then apply $\vee_{E}$ itself to discharge the open assumptions present in the derivations of the two minor premisses. Thus, in contrast to the case of $\to$, the elimination rule is applied after the introduction rule(s) and not before. The same remark holds for the existential quantifier. One of the advantages of the Ludics framework is that we can have a homogeneous treatment of the recovery principle and we do not depend on the format of the elimination rule of the connective under analysis.} As noted in \cite{PfenningDavies2001,NP2015}, under the proofs-as-programs correspondence, this transformation can be seen as an $\eta$-expansion step in typed $\lambda$-calculus.

When taken together, the inversion principle and the recovery principle guarantee a balance between the introduction and the elimination rules of a certain connective: the elimination rules are no more and no less informative than the introduction rules. Consequently, by borrowing a terminology introduced by Dummett, some authors say that the rules which satisfy both the inversion and the recovery principle are \textit{harmonious} (see \cite{FrancezDyckhoff2012}). Moreover, as we already mentioned, asking for both the inversion and the recovery principle is a way of demanding that the meaning of a connective is completely determined by its inference rules (and more specifically, by its introduction rules). This means that to fix the meaning of such of a connective, we do not need to look for any specific context of where to fix the reference (or the denotation) of this connective. Traditionally, a linguistic expression whose meaning is independent, and thus invariant, from any referential (or denotational) context is usually identified with what we call a \textit{logical constant}. It is for this reason that harmony is considered to play not only the role of a meaning criterion but also of a logicality criterion (see \cite[pp.~286--287]{Dummett91}).

We propose here a way to clarify both of (1) the idea of ``containment" present in Prawitz's inversion principle, and of (2) the idea that the recovery principle plays the role of the ``converse" of the inversion principle. These ideas are still informal, so one should provide each of them with a precise sense. For this purpose, we study the notion of harmony from the point of view of Girard's \textit{Ludics} \cite{girard2001}. More precisely, we adopt here Terui's \textit{Computational Ludics} \cite{terui2011}, because its $\lambda$-calculus-style syntax is particularly useful for our purpose. In Section \ref{harm}, we make ideas (1) and (2) more precise by reformulating the notion of harmony within Computational Ludics (see Definition \ref{harmony}). Consequently, we will \textit{generalise} the notion of connective in Computational Ludics and consider not only the ``good" (i.e. meaningful/logical) connectives satisfying our notion of harmony but also the ``bad" (i.e. non-meaningful/non-logical) connectives not satisfying it.\footnote{Our generalisation allows us to take into account connectives that do not satisfy the Ludics counterparts of the inversion principle and the recovery principle (see Example \ref{gamma} below).} In Section \ref{seccha}, we show that our notion of harmony is characterised by each of two conditions that make an essential use of the locative and interactive features proper to the Ludics approach. The first condition, which we call the \textit{dual decomposability of connectives}, is a variant of the so-called internal completeness of connectives. The second, which we call the \textit{dual decomposability of visitable paths}, is formulated in terms of the \textit{regularity} of behaviours introduced by Fouquer\'{e} and Quatrini \cite{FQ2018} in Girard's Ludics and extended by Pavaux \cite{pavaux2017,pavaux2017c} to Computational Ludics.

%%%
%%%
%%% New Section
%%%
%%%
\section{Inversion and Recovery Principles in Computational Ludics}\label{harm}
In Section \ref{dertodes}, we give the basic definitions of Terui's Computational Ludics by following \cite{terui2011,BT2010,pavaux2017c}. Next, in Section \ref{combeh}, our notion of connective is introduced. Then, in Section \ref{reformham}, we will reformulate the inversion and recovery principles as the harmony condition in Computational Ludics.

%%%
%%%
%%% New Section
%%%
%%%
\subsection{From Derivations to Computational Designs}\label{dertodes}
\textit{Computational designs} (in short, \textit{c-designs}) are the basic entities of Computational Ludics. They can be understood as abstract sequent derivations because the designs in Girard's original Ludics are such entities (for Girard's designs, see \cite{girard2001,curien2005,NaiboMattiaSeiller2016}). Consider the following procedure to extract an abstract sequent derivation from the leftmost derivation (note that the leftmost derivation contains the rule \textit{Daimon} $\maltese$, which enables one to deduce any sequent):
\begin{center}
{\footnotesize
$\infer[\with \;\; \stackrel{(1)}{\hookrightarrow} \quad]{ \vdash ( A \oplus ( B \parr C ) ) \with D }{
  \infer[\oplus]{ \vdash A \oplus ( B \parr C ) }{
    \infer[\parr]{ \vdash B \parr C }{
      \infer[\maltese]{ \vdash B , C }{}
    }
  }
  &
  \infer[\maltese]{ \vdash D }{}
}
\infer[\;\; \stackrel{(2)}{\hookrightarrow} \quad]{ ( - , ( A \oplus ( B \parr C ) ) \with D , \{ \{ A \oplus ( B \parr C ) \} , \{ D \} \} ) }{
  \infer{ ( + , A \oplus ( B \parr C ) , \{ B \parr C \} )  }{
    \infer{ ( - ,  B \parr C , \{ \{ B,C\} \} ) }{
      \maltese
    }
  }
  &
  \maltese
}
\infer{ ( - , \xi , \{ \{ 1 \} , \{ 2 \} \} ) }{
  \infer{ ( + , \xi 1 , \{ 2 \} )  }{
    \infer{ ( - ,  \xi 12 , \{ \{ 1,2 \} \} ) }{
      \maltese
    }
  }
  &
  \maltese
}$
}
\end{center}

In step (1), we encode the information on the rules' applications into triples: for example, the triple $( - , ( A \oplus ( B \parr C ) ) \with D , \{ \{ A \oplus ( B \parr C ) \} , \{ D \} \} )$ indicates that $\vdash ( A \oplus ( B \parr C ) ) \with D$ is inferred from the two premises $\vdash A \oplus ( B \parr C )$ and $\vdash D$ by a negative (i.e. reversible) rule. On the other hand, the triple $( + , A \oplus ( B \parr C ) , \{ B \parr C \} )$ indicates that $\vdash A \oplus ( B \parr C )$ is inferred from $\vdash B \parr C$ by a positive (i.e. irreversible) rule. In  step (2), we abstract the information on the positions (i.e. the \textit{locations}) of the formulas by omitting the information on their contents. The formula $( A \oplus ( B \parr C ) ) \with D$ is replaced with its location $\xi$; its first immediate subformula is denoted by its location $\xi 1$, where $1$ indicates that it is the location of the first immediate subformula, and so on. The triples and the symbol $\maltese$ in the rightmost derivation above can be considered to be actions performed to construct this derivation in the bottom-up way as in proof-search. These abstract sequent derivations are called \textit{designs as desseins} in \cite{girard2001} and \textit{untyped proofs} in \cite{NaiboMattiaSeiller2016}. These abstract sequent derivations can have infinitely long branches because one does not consider formulas anymore and so one can keep on decomposing a location $\xi$ at infinity. Moreover, they can also have infinite width (i.e. infinite branching) because there are infinitary many actions that can be applied to a location $\xi$.

C-designs can be treated as abstract sequent derivations that are expressed in the style of generalised infinite $\lambda$-terms (for the precise definition of c-designs, see Definition \ref{cdesigns} below). We will take a signature $\cA = ( A , \mathsf{ar} )$ which is a pair of a set $A$ of names and a mapping $\mathsf{ar}$ that assigns an arity to each name $a \in A$. Then, we consider positive actions and negative actions corresponding to positive triples $(+ , \xi , I )$ and negative triples $(- , \xi ', \mathcal{N} )$, respectively. In contrast to Girard's designs, negative actions in c-designs include variable binding to obtain a generalisation of $\lambda$-abstraction. However, the notion of c-design preserves a fundamental feature of Girard's Ludics: the absence of any essential distinction between syntactic and semantic level, i.e. between derivations and models (see e.g. \cite{terui2011,BT2010}). As we have just seen, c-designs are abstract sequent derivations possibly with infinite branches and infinite widths. They are thus suitable to work not only as derivations but also (counter-)models, as far as one considers infinite trees in extracting models from proof-search failure.\footnote{Note that the opposition between finite derivations and infinite (counter-)models disappears when logical systems satisfying the finite model property are considered. This is what happens, for instance, in the case of multiplicative additive linear logic $\mall$, as remarked in \cite[p.~2]{BT2010}.}

On the basis of these explanations, we provide a precise definition of c-designs. Let $\mathcal{V}$ be a countably infinite set of variables. As stated above, a \textit{signature} $\cA = (A , \mathsf{ar} )$ is a pair of a set $A$ of names and a mapping $\mathsf{ar}$ assigning an arity to each name $a \in A$. The set of \textit{positive actions} consists of Daimon $\maltese$, Divergence $\Omega$ and proper positive actions $\ov{a}$ for any $a \in A$. The set of \textit{negative actions} consists of all variables in $\mathcal{V}$ and proper negative actions $a ( x_1 ,\ldots , x_n )$ for any $a \in A$ with $\mathsf{ar}(a) = n$ and any distinct $x_1 ,\ldots , x_n \in \mathcal{V}$. We often abbreviate a proper negative action $a ( x_1 ,\ldots , x_n )$ as $a (\vec{x}_a)$. Let $\cT$ be the set of possibly non-well founded labelled trees such that (1) each of their nodes is labelled with either $\maltese$, $\Omega$, a proper positive action $\ov{a}$, a variable $x$ or an $A$-indexed set $\{ a (\vec{x}_a )\}_{a \in A}$ of proper negative actions, and (2) each of their edges is labelled with a natural number or a name.
\begin{definition}[Computational Designs]\label{cdesigns}
  The set $\cD^+$ of positive c-designs and the set $\cD^-$ of negative c-designs are defined as the largest subsets of $\cT$ satisfying the following conditions.
  \begin{itemize}
  \item If $P \in \cD^+$ holds then $(1)$ $P$ is a node labelled with $\maltese$, or $(2)$ $P$ is a node labelled with $\Omega$, or $(3)$ $P$ is of the form
      \[
      \begin{xy}
        (0,0)*+{\ov{a}}="A", (-20,10)*+{N_0}="B", (20,10)*+{N_k}="C", (0,10)*+{\cdots}="D",
        \ar @{-}"A";"B"^{0} \ar @{-}"C";"A"^{k} \ar @{-}"A";"D"
      \end{xy}
      \]
      with $\mathsf{ar}(a) = k$ and $N_0 ,\ldots ,N_k \in \cD^-$. We denote this tree by $N_0 | \ov{a} \langle N_1 ,\ldots , N_k \rangle$.

  \item If $N \in \cD^-$ holds then $(1)$ $N$ is a node labelled with a variable $x$, or $(2)$ $N$ is a tree of the form
      \[
      \begin{xy}
        (0,0)*+{\{ a (\vec{x}_a) \}_{a \in A}}="A", (-20,10)*+{\cdots}="B", (0,10)*+{P_a}="D", (20,10)*+{\cdots}="C", (30,10)*+{(a \in A)},
        \ar @{-}"A";"B" \ar @{-}"A";"C" \ar @{-}"A";"D"^a
      \end{xy}
      \]
      such that it has $\val{ A }$ immediate subtrees $\{ P_a \}_{a \in A}$ and $P_a \in \cD^+$ holds for any $a \in A$, where $\val{ A }$ denotes the cardinality of $A$. We denote this tree by $\sum a(\vec{x}_a).P_a$, and we stipulate that the variables $\vec{x}_a$ are bound in this tree.
  \end{itemize}
  Define $\cD := \cD^+ \cup \cD^-$. A subdesign of a c-design $T$ is a subtree of $T$.
\end{definition}

As explained in \cite{terui2011}, $\sum a(\vec{x}_a).P_a$ is the additive superimposition of positive c-designs $\{ P_a \}_{a \in A}$ and so, for instance, the value $\otv{ \sum a(\vec{x}_a).P_a }$ of the normal form function of $\sum a(\vec{x}_a).P_a$ is equal to $\sum a ( \vec{x}_a ) . \otv{ P_a }$ (for the definition of the normal form function, see Definition \ref{nff}). We denote positive c-designs by $P,Q$, negative c-designs by $M,N$ and positive or negative c-designs by $T,U$ possibly with suffixes. Following \cite{pavaux2017c}, we adopt Barendregt's variable condition: no variable occurs both as a free one and as a bound one in a c-design, and all bound variables in a c-design are distinct. Moreover, two $\alpha$-equivalent c-designs are identified (for the definition of $\alpha$-equivalence on c-design, see \cite[Definition 2.5]{terui2011}).

Divergence $\Omega$ allows one to express partially branching c-designs: when $K$ is a subset of $A$ and $\{ P_a \}_{a \in K}$ is a $K$-indexed family of positive c-designs, we denote by $\sum_K a ( \vec{x}_a ).P_a$ the negative c-design $\sum a ( \vec{x}_a ). Q_a$ such that $Q_a = P_a $ if $a \in K$, and $Q_a = \Omega$ otherwise. If $K$ is a finite set $\{ a_1 , \ldots , a_n \}$ we then  write $a_1 ( \vec{x}_{a_1} ). P_{a_1} + \cdots + a_n ( \vec{x}_{a_n} ). P_{a_n}$ instead of $\sum_K a ( \vec{x}_a ).P_a$. In particular, we write $a ( \vec{x}_a ). P_a$ if $K = \{ a \}$. Similar notations $\sum_{\alpha} a ( \vec{x}_a ).P_a$ are used for a set $\alpha$ of negative actions such that for any distinct $a ( \vec{x}_a ) , b ( \vec{x}_b ) \in \alpha$, $a \neq b$ holds. If we include the unary name $\lambda$ in $A$ and denote the positive action $\ov{\lambda}$ by $@$, then we have the $\lambda$-abstraction $\lambda ( x ). P$ and the $\lambda$-application $M | @ \langle N \rangle$ in Computational Ludics, so $\sum a ( \vec{x}_a ).P_a$ and $N_0 | \ov{a} \langle N_1 , \ldots , N_n \rangle$ are the generalised abstraction and the generalised application, respectively. The set of free variables in a c-design $T$ is denoted by $\fv{T}$. On the other hand, $T[N_1 / x_1 , \ldots , N_n / x_n ]$ denotes the c-design resulting from the simultaneous substitution of the negative c-design $N_i$ for all occurrences of $x_i$ in $T$ for each $i$, where bound variables in $T$ are renamed if necessary.

The $\lambda$-term-style syntax of c-designs enjoys some computational features which are similar to the ones of $\lambda$-calculus. We can define a notion of redex as a sort of $\beta$-redex called a \textit{cut}, and the execution of a redex is defined as the cut reduction: a c-design $T$ is a \textit{cut} if $T$ is a positive c-design of the form $( \sum a (\vec{x}_a ).P_a ) | \ov{a} \langle N_1 ,\ldots , N_k  \rangle$. The reduction rule for cuts is defined as $( \sum a (\vec{x}_a ).P_a ) | \ov{a} \langle \vec{ N }  \rangle \rwt P_a [ \vec{N} / \vec{x}_a ]$. A c-design is \textit{cut-free} if it has no cut. Note that, as in the case of $\lambda$-terms, a c-design can be treated both as a function and as a value. The known fact below (Theorem \ref{ass}) shows that we have a limited form of the confluence of cut reduction.

In addition, we use the following notions concerning the classification of c-designs. A variable $x$ occurring as $N_0 | \ov{a} \langle N_1 ,\ldots ,x, \ldots N_n \rangle$ in a c-design $T$ is called an \textit{identity in} $T$. A c-design $T$ is \textit{identity-free} if $T$ is not a variable and there is no identity in $T$. Intuitively, an identity in a c-design $T$ indicates that $T$ can be ``$\eta$-expanded" at the position of $x$ (for a detailed explanation of the notion of identity, see \cite[\S~2.1]{terui2011}). A c-design $T$ is \textit{total} if $T \neq \Omega$ holds. A c-design $T$ is \textit{linear} if for any of its subdesigns of the form $N_0 | \ov{a} \langle N_1 ,\ldots , N_k \rangle$, the sets $\fv{N_0} , \ldots , \fv{N_k}$ are pairwise disjoint. A c-design $T$ is \textit{standard} if $T$ is cut-free, identity-free, total, linear and $\fv{T}$ is finite.

We denote the reflexive and transitive closure of the reduction relation $\rwt$ by $\rwt^{\ast}$, and write $P \Downarrow Q$ if there is a c-design $Q$ such that $P \rwt^{\ast} Q$ holds and $Q$ is neither a cut nor $\Omega$, otherwise we write $P \Uparrow$. To define the normal form function on c-designs, we use the head normal form function $hnf$ and corecursion. Let $hnf : \cD \to \cD$ be the function preserving the polarity such that $hnf (N)  =  N$ for any negative c-design $N$, and if $P \Downarrow Q$ then $hnf (P)  =  Q$, otherwise $hnf ( P ) = \Omega$. By using the corecursive definition principle of functions on c-designs (for the proof of this principle, see \cite[\S~2.2]{terui2011}), we define the normal form function on c-designs as follows.
\begin{definition}[Normal Form Function on C-Designs]\label{nff}
  The normal form function $\otv{\cdot} : \cD \to \cD$ on c-designs are defined as follows:
  \begin{eqnarray*}
    \otv{P} & = & \maltese , \text{ if $P \Downarrow \maltese$,} \\
    & = & \Omega , \text{ if $P \Uparrow$,} \\
    & = & x | \ov{a} \langle \otv{N_1},\ldots , \otv{N_k} \rangle , \text{ if $P \Downarrow  x | \ov{a} \langle N_1 ,\ldots , N_k \rangle$,} \\
    \otv{N} & = & x, \text{ if $N = x$,} \\
    & = & \sum a (\vec{x}_a ) .\otv{P_a}, \text{ if $N = \sum a (\vec{x}_a ). P_a$.}
  \end{eqnarray*}
\end{definition}

When $\otv{T} = \maltese$ holds, we say that $T$ \textit{converges to} $\maltese$. We have the following limited version of confluence (for its proof, see \cite[Theorem 1.12]{BT2010}). It is limited in the sense that it implies the joinability only for the values of the normal form function: for example, we have $\otv{\otv{T}[N / x]} = \otv{T[\otv{N}/x]}$, which says that $\otv{T}[N / x]$ and $T[\otv{N} / x]$ are joinable with respect to their values of $\otv{\cdot}$.
\begin{theorem}[Associativity]\label{ass}
  For any c-design $T$ and any negative c-designs $N_1 , \ldots , N_n$, we have
  \[
  \otv{ T [N_1 / x_1 , \ldots , N_n / x_n ] } = \otv{ \otv{T} [ \otv{N_1} / x_1 , \ldots , \otv{N_n} / x_n ] }.
  \]
\end{theorem}

%%%
%%%
%%% New Section
%%%
%%%
\subsection{The Computational Behaviour of C-Designs}\label{combeh}
Since c-designs are untyped objects, we cannot classify them with respect to their computational behavior in advance (i.e. a priori), but we can do it a posteriori by \textit{testing} them with other c-designs. For this purpose, we first define \textit{anti-designs}. Let $x_0$ be an arbitrary but fixed variable. A positive c-design $P$ is \textit{atomic} if $\fv{P} \subseteq \{ x_0 \}$ holds. A negative c-design $N$ is \textit{atomic} if $\fv{N}$ is empty.
\begin{definition}[Anti-Designs]
  $(1)$ An anti-design against positives is a finite set $\{ (x_1 , N_1 ) ,\ldots , (x_n , N_n ) \}$ of pairs of a variable $x_i$ and an atomic negative c-design $N_i$ such that $x_1 ,\ldots , x_n$ are pairwise distinct. We say that $\{ x_1 ,\ldots , x_n \}$ is the base of this anti-design. $(2)$ An anti-design against negatives is a finite set $\{P , (x_1 , N_1 ) ,\ldots , (x_n , N_n ) \}$ such that $P$ is an atomic positive c-design and $\{ (x_1 , N_1 ) ,\ldots , (x_n , N_n ) \}$ is an anti-design against positives. We say that $\{ x_1 ,\ldots , x_n \}$ is the base of this anti-design. Later on, we denote $\{ (x_1 , N_1 ) ,\ldots , (x_n , N_n ) \}$ by $[N_1 / x_1 ,\ldots , N_n / x_n ]$, and $\{P , (x_1 , N_1 ) ,\ldots , (x_n , N_n ) \}$ by $[P, N_1 / x_1 ,\ldots , N_n / x_n ]$.
\end{definition}

The \textit{orthogonality}, which we define below, provides a way for testing c-designs with other c-designs: if a c-design $T$ is orthogonal to an anti-design against it, this means that $T$ passes the test in terms of this anti-design.
\begin{definition}[Orthogonality]
  $(1)$ A positive c-design $P$ and an anti-design $[G] = [N_1 / x_1 ,\ldots , N_n / x_n]$ against positives are orthogonal if and only if $P[N_1 / x_1 ,\ldots , N_n / x_n]$ is closed and converges to $\maltese$. $(2)$ A negative c-design $M$ and an anti-design $[G] = [P , N_1 / x_1 ,\ldots , N_n / x_n]$ against negatives are orthogonal if and only if $P[M[N_1 / x_1 ,\ldots , N_n / x_n] / x_0]$ is closed and converges to $\maltese$. When a c-design $T$ and an anti-design $[G]$ are orthogonal, we write $T \bot [G]$.
\end{definition}

An anti-design is \textit{cut-free} (resp. \textit{standard}) if any c-design contained in it is cut-free (resp. standard). When $\bd{T}$ is a set of cut-free c-designs of the same polarity and $\bd{G}$ is a set of cut-free anti-designs of the same polarity, then we define as follows:
\begin{itemize}
\item If all c-designs in $\bd{T}$ are atomic, $\bd{T}^{\bot}$ is the set of all standard and atomic c-designs $U$ with $T \bot U$ for any $T \in \bd{T}$, where for any atomic $P$ and $N$, $P \bot N :\Leftrightarrow N \bot P :\Leftrightarrow P[N / x_0]$ is closed and converges to $\maltese$. Otherwise, $\bd{T}^{\bot} := \{ [G] : \text{$[G]$ is standard and $T \bot [G]$ holds for any $T \in \bd{T}$} \}$.

\item $\bd{G}^{\bot} := \{ T : \text{$T$ is standard and $T \bot [G]$ holds for any $[G] \in \bd{G}$} \}$.
\end{itemize}
The reason why we restrict the elements of $\bd{T}^{\bot}$ and $\bd{G}^{\bot}$ to cut-free ones, in particular, standard ones is that we define \textit{behaviours} as sets of standard c-designs. Note that one may define behaviours as sets of l-designs, following \cite{terui2011}: an \textit{l-design} is an identity-free, total and linear c-design with finitely many free variables. By defining behaviour as sets of standard c-designs, i.e. sets of cut-free l-designs as in \cite{BST2010,BT2010}, one can simplify some formulations concerning behaviours and we have adopted this approach. In our case, we can make the characterisations of the harmony condition in Section \ref{seccha} simpler because a positive c-design in a behaviour is always of the form $x | \ov{a} \langle N_1 , \ldots , N_k \rangle$ due to its cut-freeness.
\begin{definition}[Behaviours]
  A set $\bd{T}$ of standard c-designs of the same polarity is a behaviour if and only if $\bd{T} = \bd{T}^{\bot\bot}$ holds. We say a behaviour $\bd{T}$ is an a-behaviour if all c-designs in $\bd{T}$ are atomic.
\end{definition}

Behaviours correspond to types inhabited by c-designs, and they give a classification of c-designs (in fact, standard c-designs) in terms of the orthogonality or tests: a standard c-design $T$ belongs to a behaviour $\bd{B}$ iff $T \bot [G]$ holds for any $[G] \in \bd{B}^{\bot}$. Note that if $\bd{B}$ is an a-behaviour, then $\bd{B}^{\bot}$ is also an a-behaviour.

Since behaviours correspond to types, a connective in Computational Ludics applies to behaviours, and returns a new behaviour. Our notion of connective is defined as follows:
\begin{definition}[Connectives]
  An $n$-\textit{ary connective} $\alpha$ is a triple $( \vec{z} ,  \alpha^I , \alpha^E )$ of a finite sequence $\vec{z}$ of variables and two finite sets of negative actions $\alpha^I , \alpha^E$ satisfying the following three conditions.
  \begin{itemize}
  \item The finite sequence $\vec{z}$ consists of $n$ distinct variables $z_1 ,\ldots , z_n$ with $x_0 \not\in \{ z_1 ,\ldots , z_n \}$.

  \item The union $\alpha^I \cup \alpha^E$ is a set $\{ a_1 ( \vec{x}_1 ) , \ldots , a_{m} ( \vec{x}_{m} ) \}$ of negative actions such that $a_1  , \ldots , a_{m}$ are pairwise distinct names and for each $i$ with $1 \leq i \leq m$, there is a natural number $k$ and indices $(i , 1) ,\ldots , (i , k)$ with $\{ \vec{ x }_i \} =  \{ x_{( i , 1)} ,\ldots , x_{(i , k)} \} \subseteq \{ \vec{z} \}$. We denote $\{ x_{(i,1)} ,\ldots , x_{(i,k)} \}$ by $X_i$.

  \item The union $\bigcup_{1 \leq i \leq m} X_i$ is equal to $\{ \vec{z} \}$.
  \end{itemize}
  We stipulate that each variable in $\vec{z}$ is bound in a connective $( \vec{z} ,  \alpha^I , \alpha^E )$.
\end{definition}

As c-designs, two $\alpha$-equivalent connectives are identified (e.g. $( x , y , \{ a ( x ) \} , \{ b ( y , x ) \} )$ is identified with $( v , z , \{ a ( v ) \} , \{ b ( z , v ) \} )$). The reason why we imposed the condition $x_0 \not\in \{ z_1 ,\ldots , z_n \}$ is that, informally speaking, we want to keep $x_0$ to be the location of atomic c-designs. This is useful for several of the formulations that follow; in particular, the formulation of game-semantic framework of \cite{pavaux2017,pavaux2017c} in Section \ref{ddv}.

A connective $\alpha = ( \vec{z} , \alpha^I , \alpha^E )$ gives the abstract information for obtaining a set of introduction and elimination rules. In general, $\alpha^I$ in a connective $\alpha = ( \vec{z} , \alpha^I , \alpha^E )$ determines the rule for $\alpha$ which constructs a c-design from an $\alpha^I$-indexed family of positive c-designs as in the leftmost tree below. On the other hand, the set $\alpha^E$ determines $\val{ \alpha^E }$ rules, where $\val{ \alpha^E }$ denotes the cardinality of the set $\alpha^E$, via the positive actions $\ov{a_1} ,\ldots , \ov{a_n}$ corresponding to the negative actions in $\alpha^E = \{ a_1 ( \vec{x}_1 ) ,\ldots , a_n ( \vec{x}_n ) \}$. These rules are expressed as the remaining trees below:
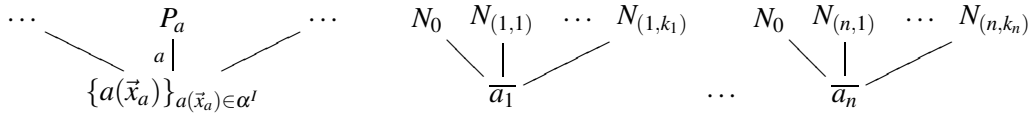
\[
\begin{xy}
  (0,0)*+{\{ a (\vec{x}_a) \}_{a ( \vec{x}_a ) \in \alpha^I}}="A", (-20,10)*+{\cdots}="B", (0,10)*+{P_a}="D", (20,10)*+{\cdots}="C",
  \ar @{-}"A";"B" \ar @{-}"A";"C" \ar @{-}"A";"D"^a
\end{xy}
\qquad
\begin{xy}
  (0,0)*+{\ov{a_1}}="A", (-10,10)*+{N_0}="B", (20,10)*+{N_{(1, k_1)}}="C", (0,10)*+{N_{(1,1)}}="D", (10,10)*+{\cdots}
  \ar @{-}"A";"B" \ar @{-}"C";"A" \ar @{-}"A";"D"
\end{xy}
\;
\cdots
\;
\begin{xy}
  (0,0)*+{\ov{a_n}}="A", (-10,10)*+{N_0}="B", (20,10)*+{N_{(n, k_n)}}="C", (0,10)*+{N_{(n,1)}}="D", (10,10)*+{\cdots}
  \ar @{-}"A";"B" \ar @{-}"C";"A" \ar @{-}"A";"D"
\end{xy}
\]
We treat the leftmost rule as the introduction rule for $\alpha$ and the other rules as the elimination rules for $\alpha$ because the c-design $( \sum_{\alpha^I} a ( \vec{x}_a ) . P_a ) | \ov{a_i} \langle N_{(i , 1)} , \ldots , N_{(i , k_i )} \rangle$ with $a_i ( \vec{x}_i ) \in \alpha^E$ is a cut; namely, a redex with respect to the reduction $\rwt$, which is a generalisation of $\beta$-reduction.

In addition, the finite sequence $\vec{z}$ in $\alpha$ fixes the arity of the connective $\alpha$. To see this, we define the \textit{semantic entailment} (introduced in \cite[Definition 2.7]{BT2010}) and \textit{behaviours composed by connectives} (a variant of the kind of behaviours defined in \cite[Definition 4.11]{terui2011}). These behaviours are also crucial  for our reformulation of harmony. A \textit{positive context} is a finite set $\{ x_1 : \bd{P}_1 , \ldots , x_n : \bd{P}_n \}$ of pairs of a variable $x_i$ and a positive a-behaviour $\bd{P}_i$ such that $x_1 , \ldots , x_n$ are pairwise distinct. A \textit{negative context} is a finite set $\{ \bd{N} \} \cup \Gam$ such that $\bd{N}$ is a negative a-behaviour and $\Gam$ is a positive context.
\begin{definition}[Semantic Entailment]
  $(1)$ Let $P$ be a positive standard c-design with $\fv{P} \subseteq \{ x_1 , \ldots , x_n \}$ and $\{ x_1 : \bd{P}_1 , \ldots , x_n : \bd{P}_n \}$ be a positive context. The entailment relation $P \models x_1 : \bd{P}_1 , \ldots , x_n : \bd{P}_n$ holds if and only if for any $M_1 \in \bd{P}_1^{\bot} , \ldots , M_n \in \bd{P}_n^{\bot}$, $P [ M_1 / x_1 , \ldots , M_n / x_n ]$ converges to $\maltese$. $(2)$ Let $N$ be a negative standard c-design with $\fv{N} \subseteq \{ x_1 , \ldots , x_n \}$ and $\{ x_1 : \bd{P}_1 , \ldots , x_n : \bd{P}_n ,\bd{N} \}$ be a negative context. The entailment relation $N \models x_1 : \bd{P}_1 , \ldots , x_n : \bd{P}_n , \bd{N}$ holds if and only if for any $M_1 \in \bd{P}_1^{\bot} , \ldots , M_n \in \bd{P}_n^{\bot}$ and any $Q \in \bd{N}^{\bot}$, $Q [ N [ M_1 / x_1 , \ldots , M_n / x_n ] / x_0 ]$ converges to $\maltese$.
\end{definition}

\begin{definition}[Behaviours Composed by Connectives]\label{BehCon}
For any name $a \in A$ with $\mathsf{ar} (a) = n$ and any negative a-behaviours $\bd{N}_1 ,\ldots , \bd{N}_n$, we define the set $\ov{a} \langle \bd{N}_1 ,\ldots , \bd{N}_n \rangle$ of negative c-designs as the set of all c-designs of the form $x_0 | \ov{a} \langle N_1 ,\ldots , N_n \rangle$ such that $N_i \in \bd{N}_i$ holds for any $i$ with $1 \leq i \leq n$.

Let $\alpha$ be an arbitrary $n$-ary connective. For any positive a-behaviours $\bd{P}_1 ,\ldots , \bd{P}_n$ and any negative a-behaviours $\bd{N}_1 ,\ldots , \bd{N}_n$, we define the positive a-behaviour $\alpha^E \langle \bd{N}_1 ,\ldots , \bd{N}_n \rangle$ and the negative a-behaviour $\alpha^I ( \bd{P}_1 ,\ldots , \bd{P}_n )$ as follows:
\begin{itemize}
\item $\alpha^E \langle \bd{N}_1 ,\ldots , \bd{N}_n \rangle := ( \bigcup_{a_i (\vec{x}_i ) \in \alpha^E} \ov{a_i} \langle \bd{N}_{(i,1)} ,\ldots , \bd{N}_{(i,k)} \rangle )^{\bot\bot}$, and
  
\item $\alpha^I ( \bd{P}_1 ,\ldots , \bd{P}_n ) := \bigcap_{a_i (\vec{x}_i ) \in \alpha^I} (\ov{a_i} \langle \bd{P}^{\bot}_{(i,1)} ,\ldots , \bd{P}^{\bot}_{(i,k)} \rangle^{\bot} )$.

\end{itemize}
\end{definition}

Note that we use $\alpha^E$ and $\alpha^I$ to define $\alpha^E \langle \bd{N}_1 ,\ldots , \bd{N}_n \rangle$ and $\alpha^I ( \bd{P}_1 ,\ldots , \bd{P}_n )$, respectively. This is the main difference between the definition above and \cite[Definition 4.11]{terui2011}.

\begin{example}\label{gamma}
Consider a ternary connective $\gamma = (x_1 , x_2 ,x_3 , \{ a (x_1 ,x_2 ) , b (x_3 )\} , \{ c (x_1 ), d (x_2 , x_3 ) \} )$. Moreover, let $\bd{P}_1 , \bd{P}_2 , \bd{P}_3$ be arbitrary positive a-behaviours. Then, if $P \models \Gam , x_1 : \bd{P}_1 , x_2 : \bd{P}_2$ and $Q \models \Gam , x_3 : \bd{P}_3$ hold, we have $a (x_1 ,x_2).P + b (x_3).Q \models \Gam , \gamma^I ( \bd{P}_1 , \bd{P}_2 , \bd{P}_3 )$. This fact corresponds to the $\gamma$-introduction rule with behaviour assignment expressed as the leftmost rule below. On the other hand, the set $\gamma^E$ gives the $\gamma$-elimination rules with behaviour assignment expressed as the remaining rules below.
\begin{center}
{\footnotesize
$\infer{a (x_1 ,x_2).P + b (x_3).Q \models \Gam , \gamma^I ( \bd{P}_1 , \bd{P}_2 , \bd{P}_3 ) }{
  P \models \Gam , x_1 : \bd{P}_1 , x_2 : \bd{P}_2
  &
  Q \models \Gam , x_3 : \bd{P}_3
}
\qquad
\infer{x | \ov{c} \langle N_1 \rangle \models \Gam , x : \gamma^E \langle \bd{P}^{\bot}_1 , \bd{P}^{\bot}_2 , \bd{P}^{\bot}_3 \rangle}{
  N_1 \models \Gam ,  \bd{P}^{\bot}_1
}
\qquad
\infer{x | \ov{d} \langle N_2 , N_3 \rangle \models \Gam ,\Del ,  x : \gamma^E \langle \bd{P}^{\bot}_1 , \bd{P}^{\bot}_2 ,\bd{P}^{\bot}_3 \rangle}{
  N_2 \models \Gam , \bd{P}^{\bot}_2
  &
  N_3 \models \Del , \bd{P}^{\bot}_3
}$
}
\end{center}
These examples show that the sequence $x_1 , x_2 , x_3$ in $\gamma$ fixes the arity of $\gamma$ and that the occurrences of $x_1 , x_2$ and $x_3$ in negative actions of $\alpha^I$ (resp. $\alpha^E$) determines the premises of the $\gamma$-introduction rule (resp. the $\gamma$-elimination rules). To illustrate the role of the order of variable-sequences in connectives, consider the connective $\delta = (x_1 ,x_2 , \{ a (x_1 ), b(x_2 ) \} , \{c (x_2 , x_1 ) \} )$, which provides the inference rules
\begin{center}
$\infer{a (x_1 ).P + b (x_2).Q \models \Gam , \bd{P}_1 \star \bd{P}_2 }{
  P \models \Gam , x_1 : \bd{P}_1
  &
  Q \models \Gam , x_2 : \bd{P}_2
}
\qquad
\infer{x | \ov{c} \langle N_2 , N_1 \rangle \models \Gam , \Del , x : \bd{P}^{\bot}_1 \: \ov{\star} \: \bd{P}^{\bot}_2}{
  N_2 \models \Gam ,  \bd{P}^{\bot}_2
  &
  N_1 \models \Del , \bd{P}^{\bot}_1
}$
\end{center}
with $\bd{P} \star \bd{Q} := \delta^I (\bd{P} , \bd{Q})$ and $\bd{N} \: \ov{\star} \: \bd{M} := \delta^E \langle \bd{N} , \bd{M} \rangle$. In the second rule, the order in which $\bd{P}_{1}^{\bot}$ and $\bd{P}_{2}^{\bot}$ appear in $\bd{P}_{1}^{\bot} \ov{\star} \bd{P}_{2}^{\bot}$ is the reverse of the way they appear in the premises, since $c (x_2 , x_1)$ has the order in which $x_2$ appears first.
\end{example}

%%%
%%%
%%% New Section
%%%
%%%
\subsection{A Reformulation of Harmony}\label{reformham}
Using our notion of connective, we define the \textit{harmony condition} in Computational Ludics.
\begin{definition}[Harmony Condition]\label{harmony}
  Let $\alpha$ be an $n$-ary connective. The connective $\alpha$ satisfies the inversion condition if and only if $\alpha^E \subseteq \alpha^I$ holds, and $\alpha$ satisfies the recovery condition if and only if $\alpha^I \subseteq \alpha^E$ holds. We say that the connective $\alpha$ satisfies the harmony condition if and only if $\alpha$ satisfies both the inversion condition and the recovery condition.
\end{definition}

Notice that the connectives $\gamma$ and $\delta$ in Example \ref{gamma} do not satisfy this harmony condition and so they are ``bad" connectives in this sense.

The inversion condition above is a reformulation of Prawitz's inversion principle in the following sense. Let $\alpha$ be an $n$-ary connective, then the inversion condition for $\alpha$ is equivalent to the following condition (see Proposition \ref{intuition}.(1) below):
\begin{enumerate}
\item[($\beta$)] for any $\alpha^I$-indexed family $\{ P_{a_j} \}_{a_j ( \vec{x}_j ) \in \alpha^I}$ of positive total c-designs, any $a_i (x_{( i , 1 )} ,\ldots , x_{( i , k )} )$ in $\alpha^E$ and any negative c-designs $N_{1} , \ldots , N_{k}$, the c-design $( \sum_{\alpha^I} a_j ( \vec{x}_j ). P_{a_j} ) | \ov{a_i} \langle N_{1} , \ldots , N_{k} \rangle$ reduces to the c-design $P_{a_i}[N_1 / x_{( i , 1 )} , \ldots , N_k / x_{( i , k )} ]$ with $P_{a_i}[N_1 / x_{( i , 1 )} , \ldots , N_k / x_{( i , k )} ] \neq \Omega$.
\end{enumerate}
This condition means that any application of an $\alpha$-elimination rule after the $\alpha$-introduction rule restores one of the subdesigns that are premises of the latter rule; that is, the $\beta$-reduction is always available for $\alpha$. Therefore, our inversion condition corresponds to Prawitz's inversion principle via its equivalence to ($\beta$). Furthermore, our inversion condition makes precise and straightforward the idea of containment: $\alpha^{E}$ is contained in $\alpha^{I}$ in the set-theoretic sense.

Next, to explain the recovery condition, the $\eta$-\textit{expanded form} of a negative c-design $N$ with respect to a connective $\alpha$ is defined as the negative c-design $\sum_{\alpha^I} a_j ( \vec{x}_j ). ( N | \ov{a_j} \langle x_{( j , 1 )} ,\ldots , x_{( j , k_j )} \rangle )$ such that $\{ x_{( j , 1 )} ,\ldots , x_{( j , k_j )}\}$ and $\fv{N}$ are disjoint for any $\alpha_j ( \vec{x}_j ) \in \alpha^I$, where bound variables in $\alpha$ are renamed if necessary. This expanded form can be depicted as the following tree:
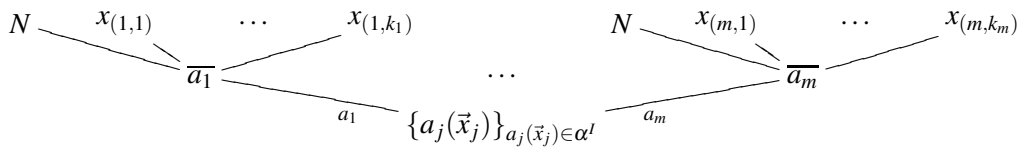
\[
\begin{xy}
  (0,0)*+{ \{ a_j ( \vec{x}_j ) \}_{ a_j ( \vec{x}_j ) \in \alpha^I } }="0", (-40,7)*+{ \ov{a_1} }="00", (-64,14)*+{ N }="000", (-50,14)*+{ x_{( 1 , 1 )}  }="001", (-33,14)*+{\cdots}="002", (-16,14)*+{  x_{( 1 , k_1 )}  }="003",
  (40,7)*+{ \ov{a_m} }="01", (16,14)*+{ N }="010", (30,14)*+{ x_{( m , 1 )} }="011", (47,14)*+{\cdots}="012", (64,14)*+{  x_{( m , k_m )} }="013",
  (0,7)*+{ \cdots }="A"
  \ar @{-}^{a_1}"0";"00" \ar @{-}"00";"000" \ar @{-}"00";"001" \ar @{-}"00";"003"
  \ar @{-}^{a_m}"01";"0" \ar @{-}"01";"010" \ar @{-}"01";"011" \ar @{-}"01";"013"
\end{xy}
\]
If we take a unary name $\lambda$ and put $@ := \ov{\lambda}$, then the $\eta$-expansion in $\lambda$-calculus can be expressed as the expansion of $N$ to $\lambda ( x ) . ( N | @ \langle x \rangle )$ and so the expanded form above is a generalisation of $\eta$-expansion in $\lambda$-calculus. Then, the following condition is equivalent to our recovery condition (see Proposition \ref{intuition}.(2) below):
\begin{enumerate}
\item[($\eta$)] there is a function $f$ mapping each negative action $a_j ( x_{( j , 1)} , \ldots , x_{(j , k)} )$ in $\alpha^I$ to a $k$-ary negative action $c_{f(j)} ( \vec{y}_{f(j)} ) \in \alpha^E$ such that for any negative c-design $N$, the $\eta$-expanded form of $N$ with respect to $\alpha$ and the c-design $\sum_{\alpha^I }a_j ( x_{(j ,1)} , \ldots , x_{(j , k)} ) . ( N | \ov{ c_{f(j)} } \langle x_{(j ,1)} , \ldots , x_{(j , k)} \rangle )$ are equal.
\end{enumerate}
This condition says that the $\eta$-expansion can be performed for $\alpha$, hence our recovery condition corresponds to the recovery principle via the condition ($\eta$) because the recovery principle means the availability of $\eta$-expansion in natural deduction. In particular, note that our recovery condition is literally the converse of the inversion condition: $\alpha^{I}$ is contained in $\alpha^{E}$. The following proposition summarises the correspondence between our inversion/recovery conditions and the inversion/recovery principles:
\begin{proposition}\label{intuition}
  Let $\alpha$ be an $n$-ary connective. $(1)$ The connective $\alpha$ satisfies the inversion condition if and only if $\alpha$ satisfies the condition $( \beta )$. $(2)$ The connective $\alpha$ satisfies the recovery condition if and only if $\alpha$ satisfies the condition $( \eta )$.
\end{proposition}
\begin{proof}
  (1) The ``only if" direction is obvious. Conversely, suppose that $\alpha^E \nsubseteq \alpha^I$ holds, and take a negative action $a_i ( \vec{x}_i )$ in $\alpha^E \setminus \alpha^I$. Moreover, let $\{ P_{a_j} \}_{a_j ( \vec{x}_j ) \in \alpha^I}$ be an $\alpha^I$-indexed family of positive total c-designs and $\vec{N}_i$ be arbitrary negative c-designs. Then, we have $( \sum_{\alpha^I} a_j ( \vec{x}_j ). P_{a_j} ) | \ov{a_i} \langle \vec{N}_i \rangle \rwt P_{a_i}[ \vec{N}_i / \vec{x}_i ] = \Omega$ because $P_{a_i} = \Omega$ holds, contradiction.

  (2) The ``only if" direction is obvious. Suppose that there is a negative action $a_j ( x_{( j , 1 )} ,\ldots , x_{( j , k )} )$ in $\alpha^I \setminus \alpha^E$. By the definition of connectives, we have $a_j \neq a_i$ for any $k$-ary negative action $a_i$ in $\alpha^E$, hence $\sum_{\alpha^I }a_j ( x_{( j , 1 )} , \ldots , x_{( j , k )} ) . ( N | \ov{c_{f(j)}} \langle x_{( j , 1 )} , \ldots , x_{( j , k )} \rangle )$ cannot be the $\eta$-expanded form. Contradiction.
\end{proof}

Our notion of connective is a generalisation of the notion of \textit{logical connective} defined in \cite{terui2011,BST2010,BT2010}. In our framework, we define that a connective $\alpha$ is \textit{logical} if $\alpha$ satisfies the harmony condition. Below we abbreviate a logical connective $( \vec{z} , \{ a_1 ( \vec{x}_1 ) ,\ldots , a_n ( \vec{x}_n ) \} , \{ a_1 ( \vec{x}_1 ) ,\ldots , a_n ( \vec{x}_n ) \} )$ as $( \vec{z} , \{ a_1 ( \vec{x}_1 ) ,\ldots , a_n ( \vec{x}_n ) \} )$. Then, it is obvious that logical connectives in our sense coincide with ones in the sense of \cite{terui2011,BST2010,BT2010}. The connectives of the linear fragment $\mallp$ of polarised linear logic are defined in \cite{terui2011} as instances of logical connectives. For example, the connective $\with$, which is called \textit{With}, can be defined as the logical connective $( x_1 , x_2 , \{ \pi_1 (x_1) , \pi_2 (x_2)\} \} )$, and this gives the following inference rules:
\[
\infer{\pi_1 ( x_1 ).P + \pi_2 ( x_2 ).Q \models \Gam , \with^I ( \bd{P}_1 , \bd{P}_2 )}{
  P \models \Gam , x_1 : \bd{P}_1
  &
  Q \models \Gam , x_2 : \bd{P}_2
}
\quad
\infer{x | \ov{\pi_1} \langle N \rangle \models \Gam ,  x : \with^E \langle \bd{P}^{\bot}_1 , \bd{P}^{\bot}_2 \rangle}{
  N \models \Gam , \bd{P}^{\bot}_1
}
\quad
\infer{x | \ov{\pi_2} \langle M \rangle \models \Gam , x : \with^E \langle \bd{P}^{\bot}_1 , \bd{P}^{\bot}_2 \rangle}{
  M \models \Gam ,  \bd{P}^{\bot}_2
}
\]
If we put $\bd{N}_1 \oplus \bd{N}_2 := \with^E \langle \bd{N}_1 , \bd{N}_2 \rangle$ as in \cite{terui2011} and write $\bd{P}_1 \with \bd{P}_2$ instead of $\with^I ( \bd{P}_1 , \bd{P}_2 )$, then the rules above are exactly the $\with$-rule and the $\oplus$-rules in one-sided sequent calculus. This is compatible with our explanations of $\with^I$ as the introduction rule and $\with^E$ as the elimination rules, because the $\with$-introduction rule corresponds to the $\with$-right rule in sequent calculus and the $\with$-elimination rules correspond to the $\oplus$-right rules via the De Morgan equivalence between $( A \with B )^{\bot}$ and $A^{\bot} \oplus B^{\bot}$. Another pair of examples from $\mallp$-connectives is the pair of lifting operators: consider the logical connective $\shneg = (x_1 , \{ \shneg (x_1) \})$, then we have the inference rules
\[
\infer{\shneg ( x_1 ).P  \models \Gam , \shneg \bd{P}}{
  P \models \Gam , x_1 : \bd{P}
}
\qquad
\infer{x | \ov{\shneg} \langle N \rangle \models \Gam ,  x : \shpos \bd{P}^{\bot}}{
  N \models \Gam , \bd{P}^{\bot}
}
\]
where $\shneg \bd{P} := \shneg^I (\bd{P})$ and $\shpos \bd{N} := \shneg^E \langle \bd{N} \rangle$.

As explained in \cite{BT2010}, a logical connective of the form
\[
(x_{(1 , 1)} , \ldots , x_{(1 , k_1)} , \ldots , x_{(m , 1)} , \ldots , x_{(m , k_m)} , \{ a_1 (x_{(1 , 1)} , \ldots , x_{(1 , k_1)}) , \ldots , a_m (x_{(m , 1)} , \ldots , x_{(m , k_m)})\})
\]
induces the following inference rules (here we suppress the term-information):
\[
\infer{\Gam , \bigwith_{1 \leq i \leq m} \bigparr_{1 \leq j \leq k_i} \bd{P}_{(i , j)}}{
  \Gam ,\bd{P}_{(1 , 1)} , \ldots , \bd{P}_{(1 , k_1)}
  &
  \cdots
  &
  \Gam ,\bd{P}_{(m , 1)} , \ldots , \bd{P}_{(m , k_m)}
}
\]
\[
\infer{\Gam_1 ,\ldots , \Gam_{k_1} , \bigoplus_{1 \leq i \leq m} \bigotimes_{1 \leq j \leq k_i} \bd{P}_{(i , j)}^{\bot}}{
  \Gam_1 , \bd{P}_{(1 , 1)}^{\bot}
  &
  \cdots
  &
  \Gam_{k_1} , \bd{P}_{(1 , k_1)}^{\bot}
}
\quad
\infer{\Gam_1 ,\ldots , \Gam_{k_m} , \bigoplus_{1 \leq i \leq m} \bigotimes_{1 \leq j \leq k_i} \bd{P}_{(i , j)}^{\bot}}{
  \Gam_1 , \bd{P}_{(m , 1)}^{\bot}
  &
  \cdots
  &
  \Gam_{k_m} , \bd{P}_{(m , k_m)}^{\bot}
}
\]
where $\bigwith_{1 \leq i \leq n}$ is the $n$-ary With and $\bigparr_{1 \leq i \leq n}$ is the $n$-ary multiplicative disjunction (Par) with their duals $\bigoplus_{1 \leq i \leq n} , \bigotimes_{1 \leq i \leq n}$. This shows that logical connectives in Computational Ludics include \textit{synthetic} connectives (\cite{curien2005}) such as the combination of $\with$ and $\parr$ (or $\oplus$ and $\otimes$) but they do not cover the combination of connectives of opposite polarities such as $\otimes$ and $\parr$. Therefore, Acclavio-Maieli's generalised connectives in \cite{AcclavioMaieli20} are not subsumed under logical connectives of Computational Ludics, because the former covers the combination of $\otimes$ and $\parr$. Moreover, we conjecture that Computational Ludics cannot deal with \textit{non-decomposable} logical connectives (i.e. logical connectives which cannot be decomposed into standard $\mall$ connectives), which are studied in \cite{AcclavioMaieli20}. On the other hand, logical connectives of Computational Ludics subsume additive connectives, synthetic connectives and the units $\bot , \top , \bd{0} , \bd{1}$, and we conjecture that the framework of \cite{AcclavioMaieli20} cannot deal with these connectives and units.\footnote{As to our connectives not satisfying the harmony condition, they have the following difference from Acclavio-Maieli's generalised connectives: the latter connectives always satisfy a form of the inversion principle, because these connectives are designed to satisfy cut reduction, which is indeed nothing but the sequent calculus counterpart of detour reduction in natural deduction.}

While logical connectives are the connectives satisfying the harmony condition, there are inharmonious connectives, as we have seen some of them in Example \ref{gamma}. The failure of the harmony condition can be sometimes tricky. Consider a connective $\alpha_{0} = (x_{1}, x_{2},\{a(x_{1}), b(x_{2})\}, \{c(x_{1}), b(x_{2})\})$. Neither the inversion condition nor the recovery condition are satisfied by $\alpha_{0}$, but $\alpha_0$ has the negative action $b ( x_2 )$ which is a common element of both $\alpha_{0}^{I}$ and $\alpha_{0}^{E}$. This means that some $\beta$-reduction steps are definable for $\alpha_{0}$, although not all of the $\beta$-reduction steps are. Indeed, for any positive c-designs $P$ and $Q$ with $P \neq \Omega \neq Q$, we have $a ( x_1 ) .P + b ( x_2 ) .Q | \ov{b} \langle N \rangle \rwt Q[N / x_2  ] \neq \Omega$ but $a ( x_1 ) .P + b ( x_2 ) .Q | \ov{c} \langle N \rangle \rwt \Omega$. Hence, though the connective $\alpha_0$ is not a logical one, it is a connective that is not deprived of any meaning. More precisely, it is not completely deprived of any computational meaning, as it allows some $\beta$-reduction steps.

%%%
%%%
%%% New Section
%%%
%%%
\section{Two Characterisations of Harmony Condition}\label{seccha}
In this section, we give two conditions each of which is equivalent to the harmony condition defined in the previous section. This will show that the harmony condition is in fact equivalent to a form of completeness which is proper to the Ludics point of view: the absence of any fundamental distinction between derivations and models is what makes it possible to pass from one to another, and vice versa (as to this viewpoint, see Section \ref{dertodes}). In Section \ref{ddc}, we propose the first condition called the \textit{dual decomposability of connectives}. Informally, a connective $\alpha$ is dually decomposable if the ``introduction" behaviour $\alpha^I ( \bd{P}_1 ,\ldots , \bd{P}_n )$ (resp. the ``elimination" behaviour $\alpha^E \langle \bd{N}_1 ,\ldots , \bd{N}_n \rangle$) is decomposed into $\bd{P}_1 ,\ldots , \bd{P}_n$ (resp. $\bd{N}_1 ,\ldots , \bd{N}_n$) via the negative actions in $\alpha^E$ (resp. $\alpha^I$). Specifically, we will have $\alpha^E \langle \bd{N}_1 ,\ldots , \bd{N}_n \rangle = \bigcup_{a_i (\vec{x}_i ) \in \alpha^I} \ov{a_i} \langle \bd{N}_{(i,1)} ,\ldots , \bd{N}_{(i,k)} \rangle \cup \{ \maltese \}$ in the case of the ``elimination" behaviour. Following \cite[p. 409]{girard2001}, one can see a form of completeness here: if the set $\bd{E} = \bigcup_{a_i (\vec{x}_i ) \in \alpha^E} \ov{a_i} \langle \bd{N}_{(i,1)} ,\ldots , \bd{N}_{(i,k)} \rangle$ is treated as a set of derivations composed from $\bd{N}_1 ,\ldots , \bd{N}_n$ by $\alpha^E$, $\bd{E}^{\bot}$ serves as a set of models which are orthogonal to any derivation in $\bd{E}$. Then, the biorthogonal $\bd{E}^{\bot\bot}$ corresponds to the set of derivations validated by these models, and $\bd{E}^{\bot\bot}$ is equal to $\alpha^E \langle \bd{N}_1 ,\ldots , \bd{N}_n \rangle$ by definition. The dual decomposability of connectives says that the derivations in $\bd{E}^{\bot\bot}$ except $\maltese$ are already included in $\bigcup_{a_i (\vec{x}_i ) \in \alpha^I} \ov{a_i} \langle \bd{N}_{(i,1)} ,\ldots , \bd{N}_{(i,k)} \rangle$, which is a set of derivations composed from $\bd{N}_1 ,\ldots , \bd{N}_n$ by $\alpha^I$. In this sense, $\bigcup_{a_i (\vec{x}_i ) \in \alpha^I} \ov{a_i} \langle \bd{N}_{(i,1)} ,\ldots , \bd{N}_{(i,k)} \rangle$ is ``complete".

In Section \ref{ddv}, we propose the second condition which is equivalent to the harmony condition. We call this condition the \textit{dual decomposability of visitable paths}. A visitable path is a sequence of actions induced by the interaction between the elements of an orthogonal pair of a c-design and an anti-design. In other words, a visitable path is an observable trace in the interaction between some c-design and anti-design (i.e. some programs), even if these programs are treated as black boxes (for a detailed discussion on the observability in Ludics, see \cite{faggian2006}). As remarked in \cite{FQ2018,pavaux2017,pavaux2017c}, visitable paths are closely related to the notion of interaction in game semantics. The dual decomposability of visitable paths says that one can find the decomposability and the completeness mentioned above not only in c-designs but also in these observable traces. Consider, for instance, a visitable path induced by some c-design $T$ in $\alpha^E \langle \bd{N}_1 ,\ldots , \bd{N}_n \rangle$ and some anti-design against $T$. Typically, such a visitable path has an action $x_0 | \ov{a_i} \langle \vec{x}_i \rangle$ with $a_i (\vec{x}_i ) \in \alpha^I$ as its first element, and the remaining sequence is obtained by ``shuffling" some visitable paths in $\bd{N}_{(i , 1)} , \ldots , \bd{N}_{(i , k)}$.

%%%
%%%
%%% New Section
%%%
%%%
\subsection{Dual Decomposability of Connectives}\label{ddc}
To formulate the intuition behind the dual decomposability of connectives precisely, we define \textit{counter sets} by adapting the definition of $\ov{\alpha}^c \langle \bd{N}_1 ,\ldots , \bd{N}_n \rangle$ and $\alpha^c ( \bd{P}_1 ,\ldots , \bd{P}_n )$ in \cite[p.~2068]{terui2011} to our setting.
\begin{definition}[Counter Sets]
  Let $\alpha$ be an $n$-ary connective. For any positive a-behaviours $\bd{P}_1 ,\ldots , \bd{P}_n$ and any negative a-behaviours $\bd{N}_1 ,\ldots , \bd{N}_n$, we define the counter set $\alpha^I ( \bd{P}_1 ,\ldots , \bd{P}_n )^{\sC}$ for $\alpha^I ( \bd{P}_1 ,\ldots , \bd{P}_n )$ and the counter set $\alpha^E \langle \bd{N}_1 ,\ldots , \bd{N}_n \rangle^{\sC}$ for $\alpha^E \langle \bd{N}_1 ,\ldots , \bd{N}_n \rangle$ as follows.
  \begin{itemize}
  \item $\alpha^I ( \bd{P}_1 ,\ldots , \bd{P}_n )^{\sC} := \bigcup_{a_i (\vec{x}_i ) \in \alpha^I} \ov{a_i} \langle \bd{P}^{\bot}_{(i,1)} ,\ldots , \bd{P}^{\bot}_{(i,k)} \rangle$, and

  \item $\alpha^E \langle \bd{N}_1 ,\ldots , \bd{N}_n \rangle^{\sC}$ is defined as the set of all negative c-designs $N$ of the following form: for some $a_i ( \vec{x}_i ) \in \alpha^E$, some $x_{(i,l)} \in \{ \vec{x}_i \}$ and some $Q \in \bd{N}^{\bot}_{(i,l)}$,
\[
N = a_i ( \vec{x}_i ) . Q [x_{(i,l)} / x_0 ] + b_1 ( \vec{x}_{b_1} ). \maltese + \cdots + b_{m} ( \vec{x}_{b_m} ).\maltese ,
\]
where $\alpha^E \setminus \{ a_i ( \vec{x}_i ) \} = \{ b_1 ( \vec{x}_{b_1} ),\ldots , b_m ( \vec{x}_{b_m} )\}$ holds. Below we use the following abbreviation: $a_i ( \vec{x}_i ) . Q [x_{(i,l)} / x_0 ] +  \maltese_{\alpha^E} := a_i ( \vec{x}_i ) . Q [x_{(i,l)} / x_0 ] + b_1 ( \vec{x}_{b_1} ). \maltese + \cdots + b_{m} ( \vec{x}_{b_m} ).\maltese$.
  \end{itemize}
\end{definition}
%Note that, in the definition above, we use $\alpha^I$ and $\alpha^E$ to define $\alpha^I ( \bd{P}_1 ,\ldots , \bd{P}_n )^{\sD}$ and $\alpha^E \langle \bd{N}_1 ,\ldots , \bd{N}_n \rangle^{\sD}$, respectively. This is because $\alpha^I ( \bd{P}_1 ,\ldots , \bd{P}_n )^{\sD}$ (resp. $\alpha^E \langle \bd{N}_1 ,\ldots , \bd{N}_n \rangle^{\sD}$) is intended as the dual of $\alpha^I ( \bd{P}_1 ,\ldots , \bd{P}_n )$ (resp. $\alpha^E \langle \bd{N}_1 ,\ldots , \bd{N}_n \rangle$), which is defined by using $\alpha^I$ (resp. $\alpha^E$).

For any negative a-behaviours $\bd{N}_1 ,\ldots , \bd{N}_n$, we define
\[
  [\bd{N}_1 /x_1 ,\ldots , \bd{N}_n /x_n ] := \{ [N_1 /x_1 ,\ldots , N_n /x_n ] : \text{$N_i \in \bd{N}_i$ for any $i$ with $1 \leq i \leq n$} \}.
\]
\begin{definition}[Dual Decomposability of Connectives]
  Let $\alpha$ be an $n$-ary connective. Then, $\alpha$ is dually decomposable if and only if $\alpha$ satisfies the following conditions:
  \begin{enumerate}
  \item $\alpha^E \langle \bd{N}_1 ,\ldots , \bd{N}_n \rangle = \alpha^I ( \bd{N}^{\bot}_1 ,\ldots , \bd{N}^{\bot}_n )^{\sC} \cup \{ \maltese \}$, and

  \item $\sum a ( \vec{x}_a ).P_a \in \alpha^I ( \bd{P}_1 ,\ldots , \bd{P}_n )$ holds if and only if $P_{a_i} \in [\bd{P}^{\bot}_{(i,1)} / x_{(i,1)} ,\ldots , \bd{P}^{\bot}_{(i,k)} / x_{(i,k)} ]^{\bot}$ holds for any $a_i ( x_{(i,1)} ,\ldots , x_{(i,k)} )$ in $\alpha^E$.
  \end{enumerate}
\end{definition}

For any logical connective $\alpha = ( \vec{z} ,\alpha_0 )$, the dual decomposability of $\alpha$ is essentially equivalent to the internal completeness of $\alpha$ formulated in \cite[\S~4.4]{terui2011} because $\alpha^I ( \bd{P}_1 ,\ldots , \bd{P}_n )^{\sC}$ and $\alpha^E \langle \bd{N}_1 ,\ldots , \bd{N}_n \rangle^{\sC}$ are equal to $\ov{\alpha}^c \langle \bd{P}^{\bot}_1 ,\ldots , \bd{P}^{\bot}_n \rangle$ and $\alpha^c ( \bd{N}^{\bot}_1 ,\ldots , \bd{N}^{\bot}_n )$ in \cite{terui2011}, respectively. By using this equivalence, one can prove Lemma \ref{internal} and Proposition \ref{proone} below in a manner similar to the proof of \cite[Lemma 4.13, Theorem 4.14]{terui2011} because any connective satisfying the harmony condition is a logical connective. We will prove Lemma \ref{internal} for readers' convenience, but omit a proof of Proposition \ref{proone}. Note that in \cite{terui2011}, c-designs in behaviours may include cuts and so head normal c-designs are used in \cite[Lemma 4.13]{terui2011}. Here, instead, any c-design in behaviours is cut-free; we thus need not use head normal c-designs.
\begin{lemma}\label{internal}
  Let $\alpha = (\vec{z} , \alpha_0 )$ be an $n$-ary logical connective. We have the following assertions:
  \begin{enumerate}
  \item $(\alpha_0 \langle \bd{N}_{1} ,\ldots , \bd{N}_{n} \rangle^{\sC} )^{\bot}  \subseteq ( \bigcup_{a_i (\vec{x}_i) \in \alpha_0} \ov{a_i}\langle \bd{N}_{(i,1)} ,\ldots , \bd{N}_{(i,k)} \rangle ) \cup \{ \maltese \}$.

  \item $\alpha_0 \langle \bd{N}_{1} ,\ldots , \bd{N}_{n} \rangle \subseteq (\alpha_0 \langle \bd{N}_{1} ,\ldots , \bd{N}_{n} \rangle^{\sC})^{\bot}$.

  \item If $\sum a ( \vec{x}_a ).P_a \in (\alpha_0 ( \bd{P}_{1} ,\ldots , \bd{P}_{n} )^{\sC} )^{\bot}$ holds then for any $a_i ( x_{(i,1)} ,\ldots , x_{(i,k)} ) \in \alpha_0$, we have
    \[
    P_{a_i} \in [\bd{P}^{\bot}_{(i,1)} / x_{(i,1)} ,\ldots , \bd{P}^{\bot}_{(i,k)} / x_{(i,k)} ]^{\bot}.
    \]

  \item $\alpha_0 ( \bd{P}_{1} ,\ldots , \bd{P}_{n} ) = (\alpha_0 ( \bd{P}_{1} ,\ldots , \bd{P}_{n} )^{\sC} )^{\bot}$.
    
  \end{enumerate}
\end{lemma}
\begin{proof}
  (1.) Assume that $P \in (\alpha_0 \langle \bd{N}_{1} ,\ldots , \bd{N}_{n} \rangle^{\sC} )^{\bot}$ holds. The case of $P = \maltese$ is trivial, so let $P$ be $x_0 | \ov{b} \langle M_1 ,\ldots , M_m \rangle$. By assumption, $b = a_i$ holds for some $a_i ( \vec{x}_i ) \in \alpha_0$. Fix an arbitrary $x_{(i , l)}$ and an arbitrary $Q \in \bd{N}_{(i,l)}^{\bot}$, then we have $P \bot a_i ( \vec{x}_i ) . Q [x_{(i,l)} / x_0 ] +  \maltese_{\alpha_0}$. Therefore, $Q \bot M_l$ holds for any $Q \in \bd{N}_{(i,l)}^{\bot}$. It follows that each $M_l$ belongs to $\bd{N}_{(i,l)}$, hence $P \in \bigcup_{a_i (\vec{x}_i) \in \alpha_0} \ov{a_i}\langle \bd{N}_{(i,1)} ,\ldots , \bd{N}_{(i,k)} \rangle$ holds.

  (2.) Consider $P \in \alpha_0 \langle \bd{N}_{1} ,\ldots , \bd{N}_{n} \rangle$ and $N = a_i ( \vec{x}_i ) . Q [x_{(i,l)} / x_0 ] +  \maltese_{\alpha_0}\in \alpha_0 \langle \bd{N}_{1} ,\ldots , \bd{N}_{n} \rangle^{\sC}$. We show $P \bot N$, and it suffices to verify that $N \in ( \bigcup_{a_i (\vec{x}_i) \in \alpha_0} \ov{a_i}\langle \bd{N}_{(i,1)} ,\ldots , \bd{N}_{(i,k)} \rangle )^{\bot}$ holds. This holds by the definition of $\alpha_0 \langle \bd{N}_{1} ,\ldots , \bd{N}_{n} \rangle^{\sC}$.

  (3.) Assume that $\sum a ( \vec{x}_a ).P_a \in (\alpha_0 ( \bd{P}_{1} ,\ldots , \bd{P}_{n} )^{\sC} )^{\bot}$ holds, and consider an arbitrary $a_i ( \vec{x}_i ) \in \alpha_0$. Then, for any $N_1 \in \bd{P}_{(i,1)}^{\bot} ,\ldots , N_k \in \bd{P}_{(i,k)}^{\bot}$, we have $Q := x_0 | \ov{a_i} \langle N_1 , \ldots , N_k \rangle \in \ov{a_i} \langle \bd{P}_{(i,1)}^{\bot} ,\ldots , \bd{P}_{(i,k)}^{\bot} \rangle$, hence $\sum a ( \vec{x}_a ).P_a$ and $Q$ are orthogonal. Therefore, we have
  \[
  P_{a_i} [N_1 / x_{(i,1)} ,\ldots , N_k / x_{(i,k)} ] \rwt \maltese
  \]
  and so $P_{a_i} \in [\bd{P}^{\bot}_{(i,1)} / x_{(i,1)} ,\ldots , \bd{P}^{\bot}_{(i,k)} / x_{(i,k)} ]^{\bot}$ holds.

  (4.) This follows from $\bigcap_{a_i (\vec{x}_i ) \in \alpha_0} ( \ov{a_i} \langle \bd{P}^{\bot}_{(i,1)} ,\ldots , \bd{P}^{\bot}_{(i,k)} \rangle^{\bot} ) = (\bigcup_{a_i (\vec{x}_i ) \in \alpha_0} \ov{a_i} \langle \bd{P}^{\bot}_{(i,1)} ,\ldots , \bd{P}^{\bot}_{(i,k)} \rangle )^{\bot}$.
\end{proof}

By the lemma above, we have the following proposition. The converse of this proposition will be obtained by Proposition \ref{protwo} below, which finishes not only the first characterisation of the harmony condition but also the second.
\begin{proposition}\label{proone}
  If a connective $\alpha$ satisfies the harmony condition, then $\alpha$ is dually decomposable.
\end{proposition}

%%%
%%%
%%% New Section
%%%
%%%
\subsection{Dual Decomposability of Visitable Paths}\label{ddv}
In this subsection, we discuss our second characterisation of the harmony condition. There are three notions that are crucial for this characterisation: \textit{interaction sequences}, \textit{visitable paths} and the \textit{regularity} of behaviours. We follow \cite{pavaux2017,pavaux2017c} in defining these notions. We first introduce several notions concerning sequences of actions, in order to define interaction sequences.
\begin{definition}[Located Actions]
  A located action is one of the following expressions: $(1)$ Daimon $\maltese$, $(2)$ an expression $x | \ov{a} \langle x_1 ,\ldots , x_n \rangle$ containing a variable $x$ and a proper positive action $\ov{a}$ followed by $x_1 ,\ldots , x_n$ such that $\mathsf{ar}(a) = n$ holds and $x , x_1 ,\ldots , x_n $ are pairwise distinct, $(3)$ an expression $a^x (x_1 ,\ldots , x_n)$ which consists of a variable $x$ and a proper negative action $a (x_1 ,\ldots , x_n)$ with $x \not\in \{ x_1 ,\ldots , x_n\}$.
\end{definition}

In the definition above, we made a slight modification of the notations in \cite{pavaux2017,pavaux2017c}: we denote negative located actions by $a^x (x_1 ,\ldots , x_n)$ instead of $a_x (x_1 ,\ldots , x_n)$. The empty sequence is denoted by $\epsilon$. In addition, we use the following variables: $\kappa$ for located actions, $\kappa^+$ for positive located actions and $\kappa^-$ for negative located actions. Hereafter, the word ``actions" always means located actions. When $\kappa$ is of the form $x | \ov{a} \langle x_1 ,\ldots , x_n \rangle$ or $a^x (x_1 ,\ldots , x_n)$, we say $a$ is the \textit{name} of $\kappa$, $x$ is the \textit{address} of $\kappa$ and $x_1 ,\ldots ,x_n$ are the \textit{arguments} of $\kappa$. Located actions except $\maltese$ are called \textit{proper located actions}.

The basic entities in this subsection are the following sequences of actions.
\begin{definition}[Alternated Justified Sequences]
  A finite sequence $s = \kappa_1 \cdots \kappa_n$ of actions is an alternated justified sequence $($in short, aj-sequence$)$ if and only if $s$ satisfies all of the following conditions:
  \begin{itemize}
  \item Alternation: the polarity of $\kappa_i$ is the opposite of the polarity of $\kappa_{i+1}$ for any $i$ with $1 \leq i \leq n$.

  %\item Variables 1: for any distinct $\kappa_i$ and $\kappa_j$, there is no variable which belongs to both the arguments of $\kappa_i$ and the argument of $\kappa_j$.

  %\item Variables 2: for any $\kappa_i ,\kappa_j$ with $j < i$ and any argument $y$ of $\kappa_i$, $y$ is not the address of $\kappa_j$.

  \item Linearity: each variable occurring in $s$ is the address of at most one action in $s$.

  \item Daimon: if $\maltese$ appears in $s$ then $\maltese = \kappa_n$ holds.

  \item Justification: for any proper action $\kappa_i$ in $s$, either $(1)$ there is a unique action $\kappa_j$ of the opposite polarity such that $j < i$ holds and the arguments of $\kappa_j$ includes the address of $\kappa_i$ or $(2)$ there is no $\kappa_j$ in $s$ such that the arguments of $\kappa_j$ includes the address of $\kappa_i$. We say that $\kappa_i$ is justified by $\kappa_j$ and denote $\kappa_j$ by $\mathsf{just}( \kappa_i )$ if $(1)$ holds, otherwise we say $\kappa_i$ is initial.

  \end{itemize}
  We say that $x$ is free in an aj-sequence $s$ if and only if $x$ occurs in $s$ only as the address of some action in $s$, and that $x$ is bound in $s$ if and only if $x$ occurs in $s$ as an argument of some action in $s$.
\end{definition}

Note that the empty sequence $\epsilon$ is trivially an aj-sequence and that we adopted Barendregt's variable convention (cf. \S~\ref{harm}). We identify two aj-sequences $s_1$ and $s_2$ that are identical modulo renaming of bound variables of $s_1$ and $s_2$. For example, $( x | \ov{a} \langle y_1 , y_2 \rangle )(a_{y_1} (z_1 , z_2 ) )$ and $( x | \ov{a} \langle v_1 , v_2 \rangle )(a_{v_1} (w_1 , w_2 ) )$ are the same aj-sequence. If $\kappa = x | \ov{a}\langle x_1 ,\ldots , x_n \rangle$ (resp. $\kappa = a^x ( x_1 ,\ldots , x_n ) $) holds, we write $a^x ( x_1 ,\ldots , x_n )$ (resp. $x | \ov{a} \langle x_1 ,\ldots , x_n \rangle$) as $\ov{\kappa}$. If $s = \kappa_1 \cdots \kappa_n$ is a non-empty sequence of proper actions, we denote $\ov{\kappa_1} \cdots \ov{\kappa_n}$ by $\ov{s}$. Moreover, we put $\ov{\epsilon} := \epsilon$. Let $s = \kappa_1 \cdots \kappa_n$ be a finite sequence of actions ($n \geq 0$) such that $\kappa_n$ is the only occurrence of $\maltese$ if $\maltese$ occurs in $s$. We define the \textit{dual} $\dual{s}$ of $s$ as follows: if $\maltese$ occurs in $s$ then $\dual{s} := \ov{\kappa_1} \cdots \ov{\kappa_{n-1}}$, otherwise $\dual{s} := \ov{ s } \maltese$. We in particular have $\dual{ \maltese } = \epsilon$, $\dual{ \epsilon } = \maltese$ and $\dual{ \dual{s} } = s$.

Next, we define \textit{paths}, which subsume some interaction sequences as typical examples. But we first define \textit{views} and \textit{anti-views} of aj-sequences to introduce the notion of path.
\begin{definition}[Views and Anti-Views of Alternated Justified Sequences]\label{def:views}
  Let $s$ be an aj-sequence. We define the view $\view{s}$ of $s$ by induction. $(1)$ If $s = \epsilon$ holds, then $\view{s} := \epsilon$. $(2)$ If $s = s'\kappa^+$ holds, then $\view{s} := \view{s'}\kappa^+$. $(3)$ Let $s$ be $s' \kappa^-$. If $\kappa^-$ is initial, then $\view{s} := \kappa^-$, otherwise $\view{s} := \view{s_0} \kappa^-$ where $s_0$ is the prefix of $s$ such that $\kappa^-$ is justified by the last action of $s_0$.

  The anti-view $\antiview{s}$ of $s$ is defined as $\antiview{s} := \dual{ \view{ s_0 } }$ with $s_0 = \dual{s}$.
\end{definition}

%When $\kappa_i$ and $\kappa_j$ occur in an aj-sequence $s$, we say $\kappa_i$ is \textit{hereditarily justified} by $\kappa_j$ in $s$ if and only if there are actions $\kappa_{1} ,\ldots ,\kappa_{n}$ with $n \geq 1$ in an aj-sequence $s$ such that $\kappa_i = \kappa_n$ holds, $\kappa_j = \kappa_1$ holds and $\kappa_{l+1}$ is justified by $\kappa_l$ for any $l$ with $1 \leq l \leq n$.
\begin{definition}[Paths]\label{def:paths}
  A path is an aj-sequence $s$ satisfying 1. and 2. below.
  \begin{enumerate}
  \item Proponent-visibility: For any prefix $s_0 \kappa^+$ of $s$ with $\kappa^+$ proper, if $\kappa^+$ is justified in $s_0$ then $\mathsf{just}( \kappa^+ )$ occurs in $\view{s_0}$.

  \item Opponent-visibility: For any prefix $s_0 \kappa^-$ of $s$, if $\kappa^-$ is justified in $s_0$ then $\mathsf{just}( \kappa^- )$ occurs in $\antiview{s_0}$.
  \end{enumerate}
\end{definition}

A non-empty path is called \textit{positive} (resp. \textit{negative}) if its first action is positive (resp. negative), and the empty path $\epsilon$ is defined as a negative path. When $D$ is a set of sequences of actions and $\kappa$ is a proper action, we denote the set $\{ \kappa  s : s \in D \}$ of sequences by $\kappa D$.

One can also consider views and paths occurring in c-designs or anti-designs, where c-designs and anti-designs are treated as trees or forests formed by views (see Figure \ref{figone}). These notions of views and paths are used in our proof too. 
\begin{definition}[Views and Paths of C-Designs and Anti-Designs]
  Let $P$ be a positive c-design and $N$ be a negative c-design with $x \not\in \fv{N}$, and assume that both $P$ and $N$ are cut- and identity-free. We define the two sets $\viewd{P}$ and $\viewd{N}_x$ of sequences of actions simultaneously:
  \begin{itemize}
  \item $\viewd{\Omega} := \emptyset$, $\viewd{\maltese} := \{ \maltese \}$ and $\viewd{ x | \ov{a} \langle \vec{ N } \rangle } := \{ \kappa^+_a \} \cup \bigcup_{i \leq \mathsf{ar}(a)}  \kappa^+_a \viewd{ N_i }_{y_i}$, where $\vec{y} = y_1 , \ldots , y_k$ are fresh and $\kappa^+_a = x | \ov{a} \langle \vec{y} \rangle$ holds,

  \item $\viewd{ N } := \viewd{ N }_{x_0}$, $\viewd{ \sum a ( \vec{x}_a ). P_a }_x := \{ \epsilon \} \cup \bigcup \{ \kappa^-_a \viewd{P_a} \cup \{ \kappa^-_a \}  : P_a \neq \Omega \}$ with $\kappa^-_a = a^x ( \vec{x}_a )$.

  \end{itemize}
  A sequence $s$ of actions is a view of $P$ $($resp. a view of $[N /x])$ if and only if $s \in \viewd{P}$ $($resp. $s \in \viewd{N}_x)$ holds. A path $p$ is a path of $P$ $($resp. a path of $[N/x])$ if and only if for any non-empty prefix $s$ $($resp. any prefix $s)$ of $p$, $\view{s}$ is a view of $P$ $($resp. a view of $[N/x])$.

  A view of a cut-free anti-design $[G]$ is a view of some member of $[G]$. A path of a cut-free anti-design $[G]$ against positives $($resp. a cut-free anti-design $[G]$ against negatives$)$ is a negative path $($resp. a positive path$)$ $s$ such that for any prefix $s_0$ $($resp. any non-empty prefix $s_0 )$ of $s$, $\view{s_0}$ is a view of $[G]$.
\end{definition}

Intuitively, a view of a c-design $T$ is a branch (or one of its prefixes) of Pavaux's tree-representation $\mathcal{T} ( T )$ of the c-design $T$ (for the details, see \cite[\S~3.1]{pavaux2017} and \cite[\S~1.2]{pavaux2017c}). Then, a path of $T$ is a sequence in $\mathcal{T} ( T )$ traced by proceeding along possibly several branches from the root. For instance, in this representation of c-designs, the positive c-design $P$ and the anti-design $[N / x_0]$ with
\[
P = x_0 | \ov{a} \Bigl\langle b ( x_1 ) . ( x_1 | \ov{c} \langle \rangle ) , b ( x_2 ) . ( x_2 | \ov{c} \langle \rangle ) \Bigr\rangle ,\quad N = a ( y_1 , y_2 ) . \Bigl( y_1 | \ov{b} \Bigl\langle c () . \bigl( y_2 | \ov{b} \langle a ( y_5 , y_6 ) . \maltese + c () . \maltese \rangle \bigr) \Bigr\rangle \Bigr)
\]
are depicted as the left-hand tree and the right-hand tree in Figure \ref{figone}, respectively.
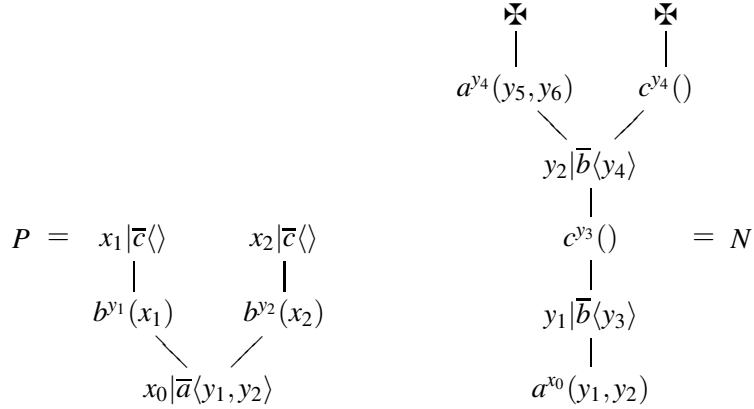
\begin{figure}[t]
\[
\begin{xy}
  (0,0)*+{x_0 | \ov{a} \langle y_1 , y_2 \rangle}="0", (-10,10)*+{b^{y_1} ( x_1 )}="00", (10,10)*+{b^{y_2} ( x_2 )}="01", (-10,20)*+{x_1 | \ov{c} \langle \rangle}="000", (10,20)*+{x_2 | \ov{c} \langle \rangle}="010", (-20,20)*+{=}, (-25,20)*+{P}
  \ar @{-}"0";"00" \ar @{-}"0";"01" \ar @{-}"00";"000" \ar @{-}"01";"010"
\end{xy}
\qquad \qquad
\begin{xy}
  (0,0)*+{a^{x_0} ( y_1 , y_2 )}="0", (0,10)*+{y_1 | \ov{b} \langle y_3 \rangle}="00", (0,20)*+{c^{y_3} ()}="000", (0,30)*+{y_2 | \ov{b} \langle y_4 \rangle}="0000", (-10,40)*+{a^{y_4} ( y_5 , y_6 )}="00000", (10,40)*+{c^{y_4} ()}="00001", (-10,50)*+{\maltese}="000000", (10,50)*+{\maltese}="000010", (15,20)*+{=}, (20,20)*+{N}
  \ar @{-}"0";"00" \ar @{-}"00";"000" \ar @{-}"000";"0000" \ar @{-}"0000";"00000" \ar @{-}"0000";"00001" \ar @{-}"00000";"000000" \ar @{-}"00001";"000010"
\end{xy}
\]
\caption{Examples of Pavaux's tree representation of c-designs}
\label{figone}
\end{figure}
The two aj-sequences
\[
x_0 | \ov{a} \langle y_1 , y_2 \rangle \;\; b^{y_1} ( x_1 ) \;\; x_1 | \ov{c} \langle \rangle \quad \text{and} \quad x_0 | \ov{a} \langle y_1 , y_2 \rangle \;\; b^{y_2} ( x_2 ) \;\; x_2 | \ov{c} \langle \rangle
\]
are views of $P$, namely the left-hand branch and the right-hand branch of $P$. The aj-sequence
\[
x_0 | \ov{a} \langle y_1 , y_2 \rangle \;\; b^{y_1} ( x_1 ) \;\; x_1 | \ov{c} \langle \rangle \;\; b^{y_2} ( x_2 ) \;\; x_2 | \ov{c} \langle \rangle
\]
is a path of $P$. Notice that views and paths of some c-design are indeed views and paths in the sense of Definitions \ref{def:views} and \ref{def:paths}.

As a further step toward the definition of interaction sequences, we define \textit{multi-designs}, which were introduced in \cite{pavaux2017c}. Multi-designs are generalisations of both c-designs and anti-designs. For the need of multi-designs in defining interaction sequences, see \cite[p.~41]{pavaux2017c}.
\begin{definition}[Multi-Designs]
  $(1)$ A negative multi-design is a finite set $\{ ( x_1 , N_1 ) , \ldots , ( x_k , N_k ) \}$ of pairs of a variable and a negative design such that $\fv{ N_1 } , \ldots , \fv{ N_k }$ are pairwise disjoint and $\fv{N_i} \cap \{ x_1 , \ldots , x_k \}$ is empty for any $i$ with $1 \leq i \leq k$. $(2)$ A positive multi-design is a finite set $\{ P , ( \vec{x} , \vec{N} ) \}$ such that $P$ is a positive design, $\{ ( \vec{x} , \vec{N} ) \} = \{ ( x_1 , N_1 ) , \ldots , ( x_k , N_k ) \}$ is a negative multi-design, $\fv{P}$ and $\fv{N_i}$ are disjoint for any $i$ with $1 \leq i \leq k$ and $\fv{P} \cap \{ x_1 , \ldots , x_k \} $ is empty.

  For any multi-design $\fD$, we define the normal form $\otv{ \fD }$ of $\fD$ as
  \[
  \otv{ \fD } := \{ ( x , \otv{N} ) : ( x , N ) \in \fD \} \cup \{ \otv{P} : P \in \fD \}.
  \]
\end{definition}

A multi-design $\fD$ is called \textit{standard} if any c-design in $\fD$ is standard. For any $( x , N )$ in some multi-design $\fD$, we denote $( x , N )$ by $[N / x]$. Moreover, when $\fD = \{ [ N_1 / x_1 ] , \ldots , [ N_k / x_k ] \}$ is a negative multi-design, we denote the result of substituting $N_i$ for $x_i$ in a multi-design $\fE$ for each $i$ by $\fE [\fD]$. Note that for any positive c-design $P$, $\{ P \}$ is a multi-design and that any anti-design is a multi-design. A \textit{view} of a multi-design $\fD$ is a view of some c-design in $\fD$. A \textit{path} of a multi-design $\fD$ is a path $s$ of the same polarity as $\fD$ such that for any prefix $s_0$ of $s$, $\view{s_0}$ is a view of $\fD$. For any multi-design $\fD$, we denote the set $\bigcup_{T \in \fD} \fv{T} $ of free variables in $\fD$ by $\fv{ \fD }$, and the set $\{ x : [N / x] \in \fD \text{ for some $N$} \}$ of \textit{negative places} of $\fD$ by $\np{ \fD }$.

For example, in Figure \ref{figtwo}, the singleton $\fD$ of the rightmost tree and the set $\fE$ of the remaining trees are multi-designs.
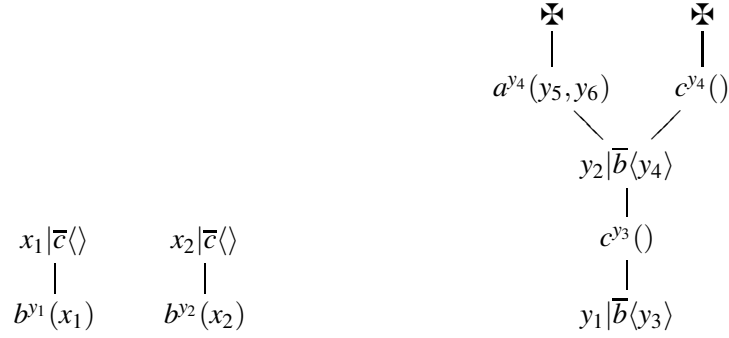
\begin{figure}[t]
\[
\begin{xy}
  (-10,10)*+{b^{y_1} ( x_1 )}="00", (10,10)*+{b^{y_2} ( x_2 )}="01", (-10,20)*+{x_1 | \ov{c} \langle \rangle}="000", (10,20)*+{x_2 | \ov{c} \langle \rangle}="010"
  \ar @{-}"00";"000" \ar @{-}"01";"010"
\end{xy}
\qquad \qquad \qquad \qquad
\begin{xy}
  (0,10)*+{y_1 | \ov{b} \langle y_3 \rangle}="00", (0,20)*+{c^{y_3} ()}="000", (0,30)*+{y_2 | \ov{b} \langle y_4 \rangle}="0000", (-10,40)*+{a^{y_4} ( y_5 , y_6 )}="00000", (10,40)*+{c^{y_4} ()}="00001", (-10,50)*+{\maltese}="000000", (10,50)*+{\maltese}="000010"
  \ar @{-}"00";"000" \ar @{-}"000";"0000" \ar @{-}"0000";"00000" \ar @{-}"0000";"00001" \ar @{-}"00000";"000000" \ar @{-}"00001";"000010"
\end{xy}
\]
\caption{Examples of Pavaux's tree representation of multi-designs}
\label{figtwo}
\end{figure}
In the $\lambda$-term-style notations, $\fD$ is $\{  y_1 | \ov{b} \langle c () . ( y_2 | \ov{b} \langle a ( y_5 , y_6 ) . \maltese + c () . \maltese \rangle ) \rangle \}$, and $\fE$ is $\{ (y_1 , b ( x_1 ) . ( x_1 | \ov{c} \langle \rangle ) ) , (y_2 , b ( x_2 ) . ( x_2 | \ov{c} \langle \rangle ) )\}$. Though $\fE$ is also an anti-design, $\fD$ is not an anti-design, because $\fD$ contains a non-atomic positive c-design. Notice that $\fD \cup \fE$ is not a multi-design, because the free variables $y_1 , y_2$ in $\fD$ belong to $\np{\fE}$.

The following definition provides some necessary conditions for the interaction between multi-designs.
\begin{definition}[Compatibility and Quasi Closed Compatibility]
  Two multi-designs $\fD$ and $\fE$ are compatible if and only if
  \begin{itemize}
  \item both $\fv{ \fD } \cap \fv{ \fE }$ and $\np{ \fD } \cap \np{ \fE }$ are empty, and

  \item either they are negative and there is a variable $x \in \np{ \fD } \cup \np{ \fE }$ such that $x \not \in \fv{ \fD } \cup \fv{ \fE }$ holds, or they are of opposite polarities.
  \end{itemize}
  Two multi-designs $\fD$ and $\fE$ are quasi closed compatible if and only if they are of opposite polarities, compatible and satisfy the condition that $\fv{ \fD } \subseteq \np{ \fE }$ and $\fv{ \fE } \subseteq \np{ \fD }$ hold.
\end{definition}

We can apply a \textit{cut} to any two compatible multi-designs. Though the notion of cut is not used in the definition of interaction sequences, this notion will be used to define the orthogonality on multi-designs and formulate Proposition \ref{asspaths} below.
\begin{definition}[Cut of Multi-Designs]
  For any two compatible multi-designs $\fD$ and $\fE$, the cut $\Cut{ \fD }{ \fE }$ of $\fD$ and $\fE$ is defined by induction on the number of elements in $\fE$: if $\fE$ is empty then we define $\Cut{ \fD }{ \fE } := \fD$. Let $\fE$ be non-empty.
    \begin{enumerate}
    \item If $P \in \fE$ holds, we put $S := \{ [M / y] \in \fD : y \in \fv{P} \}$ and define $\Cut{ \fD }{ \fE } := \Cut{ ( \fD \setminus S ) \cup \{ P[S] \} \: }{ \: \fE \setminus \{ P\} }$.
      
    \item If $[ N / x ] \in \fE$ holds, we put $S := \{ [M / y] \in \fD : y \in \fv{N} \}$ and define
      \begin{enumerate}
      \item $\Cut{ \fD }{ \fE } := \Cut{ ( \fD \setminus S ) \cup \{ [ N[S] / x ]\} \: }{ \: \fE \setminus \{ [N / x ]\} }$, if $x \not\in \fv{ \fD }$ holds,

      \item $\Cut{ \fD }{ \fE } := \Cut{ ( \fD \setminus S ) [ N[S] / x ] \: }{ \: \fE \setminus \{ [N / x ]\} }$, if $x \in \fv{ \fD }$ holds.
      \end{enumerate}
    \end{enumerate}
\end{definition}

The cut $\Cut{\fD}{\fE}$ of multi-designs $\fD , \fE$ is well-defined above because $\Cut{\fD}{\fE}$ is determined uniquely regardless of the order to apply 1, 2.(a) and 2.(b). We say that two quasi closed compatible multi-designs $\fD$ and $\fE$ are \textit{orthogonal} and write $\fD \bot \fE$ if $\maltese \in \otv{ \Cut{ \fD }{ \fE } }$ holds. Note that this definition of the orthogonality between two multi-designs is broader than the one in \cite[Definition 2.1.8]{pavaux2017c}, though this broader notion is in fact used in the proof of \cite[Proposition 2.2.12]{pavaux2017c}. The reason why we used the broader definition is that we want two multi-designs such as $\fD = \{ \maltese \}$ and $\fE = \{ [ N / x ] \}$ with $N$ closed to be orthogonal but they are not in the sense of \cite[Definition 2.1.8]{pavaux2017c}.

On the basis of the definitions above, we define the notion of interaction sequence.
\begin{definition}[Interaction Sequences]
  For any two standard multi-designs $\fD$ and $\fE$ such that they are quasi closed compatible, the interaction sequence $\iseq{ \fD }{ \fE }$ of $\fD$ with $\fE$ is defined as follows: let $P$ be the unique positive design in $\fD \cup \fE$.
  \begin{itemize}
  \item Let $P = \maltese$ be the case. We define $\iseq{ \fD }{ \fE } := \maltese$ if $P \in \fD$ holds, otherwise $\iseq{ \fD }{ \fE } := \epsilon$.

  \item If $P = \Omega$ holds then we define $\iseq{ \fD }{ \fE } := \epsilon$.

  \item Assume that $P = x | \ov{a} \langle \vec{M} \rangle$ holds. If $P \in \fD$ holds then there is a unique negative design $N$ such that $[N/x] \in \fE$ holds, otherwise there is a unique negative design $N$ such that $[N/x] \in \fD$. Let $N$ be of the form $\sum b ( \vec{y}_b ) . P_b$, and we define
    \[
    \iseq{ \fD }{ \fE } :=
    \begin{cases}
      x | \ov{a} \langle \vec{y}_a \rangle \iseq{ ( \fD \setminus \{ P\} ) \cup \{ [ \vec{M} / \vec{y}_a ] \} }{ ( \fE \setminus \{ [N / x]\} ) \cup \{ P_a \} }, & \text{if $P \in \fD$},\\
      a^x ( \vec{y}_a ) \iseq{ ( \fD \setminus \{ [N /x]\} ) \cup \{ P_a \} }{ ( \fE \setminus \{ P \} ) \cup \{ [ \vec{M} / \vec{y}_a ] \} }, & \text{else}.
    \end{cases}
    \]
  \end{itemize}
\end{definition}

During the construction of interaction sequences, one usually decomposes a multi-design. For example, consider the c-designs $P$ and $N$ in Figure \ref{figone} again. Then, the interaction sequence $\iseq{ \{ P \} }{ \{ [N / x_0] \} }$ is defined, where the first step for constructing it shortens $N$ to the rightmost tree in Figure \ref{figtwo} and decomposes $P$ into the remaining trees in Figure \ref{figtwo}. In addition, note that the path
\[
x_0 | \ov{a} \langle y_1 , y_2 \rangle \;\; b^{y_1} ( x_1 ) \;\; x_1 | \ov{c} \langle \rangle \;\; b^{y_2} ( x_2 ) \;\; x_2 | \ov{c} \langle \rangle
\]
of $P$ is equal to $\iseq{ \{ P \} }{ \{ [N / x_0] \} }$. Since $\iseq{ \fD }{ \fE } = \dual{ \iseq{ \fE }{ \fD } }$ holds for any orthogonal pair of standard multi-designs $\fD ,\fE$ (see \cite[Lemma 2.2.5]{pavaux2017c}) and the pair of $P , N$ is orthogonal, the path
\[
a^{x_0} ( y_1 , y_2 ) \;\; y_1 | \ov{b} \langle x_1 \rangle \;\; c^{x_1} ( ) \;\; y_2 | \ov{b} \langle x_2 \rangle \;\; c^{x_2} ( ) \;\; \maltese
\]
of $N$ is equal to $\dual{ \iseq{ \{ P \} }{ \{ [N / x_0] \} } } = \iseq{ \{ [N / x_0] \} }{ \{ P \} }$.

Next, we define visitable paths, which are interaction sequences induced by some orthogonal pair of a c-design and an anti-design. When $T$ is a positive c-design (resp. a negative c-design), we abbreviate $\iseq{ \{ T \} }{ [G] }$ (resp. $\iseq{ \{ ( x_0 , T )\} }{ [G] }$) as $\iseq{T}{ [G] }$, and $\iseq{ [G] }{ \{ T \} }$ (resp. $\iseq{ [G] }{ \{ (x_0 , T ) \} }$) as $\iseq{ [G] }{ T }$.
\begin{definition}[Visitable Paths]
  Let $\bd{T}$ be a set of standard c-designs of the same polarity, and $\bd{G}$ be a set of standard anti-designs of the same polarity and base.
  \begin{itemize}
  \item A path $p$ is visitable in $\bd{T}$ if and only if for some $T \in \bd{T}$ and $[G] \in \bd{T}^{\bot}$, $\iseq{T}{[G]} = p$ holds.

  \item A path $p$ is visitable in $\bd{G}$ if and only if for some $T \in \bd{G}^{\bot}$ and $[G] \in \bd{G}$, $\iseq{[G]}{T} = p$ holds.
  \end{itemize}
  We denote the set of all visitable paths in $\bd{T}$ $($resp. $\bd{G})$ by $V ( \bd{T} )$ $($resp. $V ( \bd{G} ) )$.
\end{definition}

We have the following lemma and proposition, which will be used in the proofs of Lemma \ref{mainlem}.(2) and Proposition \ref{protwo}. One can prove the assertions 1, 2 and 3 of the lemma below in the same way to Lemma 2.2.6, Lemma 2.2.10 and Lemma 3.1.5 in \cite{pavaux2017c}, respectively.
\begin{lemma}\label{formainone}
  We have the following assertions:
  \begin{enumerate}
  \item Let $\fD , \fE$ be multi-designs with $\iseq{\fD}{\fE}$ defined. If $\fD$ is positive then any non-empty prefix of $\iseq{\fD}{\fE}$ is a path of $\fD$, otherwise any prefix of $\iseq{\fD}{\fE}$ is a path of $\fD$. In particular, if $\iseq{\fD}{\fE}$ is finite, then it is a path of $\fD$.

  \item Assume that a positive multi-design $\fD$ with $\Omega \not\in \fD$ and a multi-design $\fE$ are cut-free and satisfy the following two conditions: $($i$)$ $\fD$ and $\fE$ are quasi closed compatible and have a finite interaction, and $($ii$)$ for any path $s \kappa^+$ of $\fD$ $($resp. $\fE )$ such that $\kappa^+$ is proper and $\ov{s}$ is a path of $\fE$ $($resp. $\fD )$, $\ov{ s \kappa^+ }$ is a path of $\fE$ $($resp. $\fD )$. Then, $\fD \bot \fE$ holds.

  \item Let $\bd{B}$ be an arbitrary a-behaviour. If $p \in V (\bd{B})$ holds, then for any positive-ended prefix $($resp. negative-ended prefix$)$ $s$ of $p$, we have $s \in V (\bd{B})$ $($resp. $s \maltese \in V (\bd{B}) )$.
  \end{enumerate}
\end{lemma}

When $s = \kappa_1 \cdots \kappa_n$ is a sequence of actions, a \textit{subsequence} of $s$ is a sequence $\kappa_{i_1} \cdots \kappa_{i_k}$ with $1 \leq i_1 < \cdots < i_k \leq n$, and we denote by $s \uhr s '$ the subsequence of $s$ any of whose actions occurs in a sequence $s'$. Moreover, when $p$ is a path of a multi-design $\fD$ and $\fE$ is a multi-design with $\fE \subseteq \fD$, we denote the longest subsequence of $p$ that is a path of $\fE$ by $p \uhr \fE$. For a proof of the following proposition, see \cite[Proposition 2.2.12]{pavaux2017c}.
\begin{proposition}[Associativity for Paths]\label{asspaths}
  For any cut-free multi-designs $\fD, \fE$ and $\fk{F}$ such that $\fE \cup \fk{F}$ is a multi-design with $\fE$ and $\fk{F}$ disjoint, if $\fD \bot ( \fE \cup \fk{F} )$ holds, then $\iseq{ \fE }{ \otv{ \Cut{ \fk{F} }{ \fD } } } = \iseq{ \fE \cup \fk{F} }{ \fD } \uhr \fE$ holds.
\end{proposition}

To define regular behaviours, we first define the \textit{stable} and \textit{observational orderings}, \textit{intersection} and \textit{incarnation}.
\begin{definition}[Stable Ordering and Observational Ordering]
  The stable ordering $\ssse$ on c-designs is defined as the largest binary relation on c-designs such that if $T \ssse U$ holds then one of the following conditions holds:
  \begin{enumerate}
  \item $T = \maltese = U$,

  \item $T = \Omega$ and $U \in \cD^+$,

  \item $T = N_0 | \ov{a} \langle N_1 ,\ldots , N_n \rangle$, $U = M_0 | \ov{a} \langle M_1 ,\ldots , M_n \rangle$ and $N_k \ssse M_k$ for any $k$ with $0 \leq k \leq n$,

  \item $T = x = U$,

  \item $T = \sum a (\vec{y^a}) .P_a$, $U = \sum a ( \vec{y^a} ).Q_a$ and $P_a \ssse Q_a$ for any $a \in A$.

  \end{enumerate}
  The observational ordering $\preceq$ on c-designs is defined as the largest binary relation on c-designs such that if $T \preceq U$ holds then one of 1, 2, 4 above and the following conditions 3' and 5' holds:
  \begin{enumerate}
  \item[3'.] $T = N_0 | \ov{a} \langle N_1 ,\ldots , N_n \rangle$ and either $U = M_0 | \ov{a} \langle M_1 ,\ldots , M_n \rangle$ and $N_k \preceq M_k$ for any $k$ with $0 \leq k \leq n$ or $U = \maltese$,

  \item[5'.] $T = \sum a (\vec{y^a}) .P_a$, $U = \sum a ( \vec{y^a} ).Q_a$ and $P_a \preceq Q_a$ for any $a \in A$.
  \end{enumerate}
\end{definition}

\begin{definition}[Intersection]
The intersection $T \cap U$ of c-designs $T$ and $U$ is defined by corecursion:
\begin{enumerate}
\item $\maltese \cap \maltese = \maltese$,

\item $\Omega \cap P = P \cap \Omega = \Omega$,

\item $x \mid \ov{a}\langle N_1 ,\ldots , N_k \rangle \cap x \mid \ov{a}\langle M_1 ,\ldots , M_k \rangle = x \mid \ov{a} \langle N_1 \cap M_1 ,\ldots , N_k \cap M_k \rangle$ if $N_i \cap M_i$ is defined for each $i$ with $1 \leq i \leq k$,

\item $\sum a (\vec{x}_a) .P_a \cap \sum a (\vec{x}_a) .Q_a = \sum a (\vec{x}_a) .P_a \cap Q_a$ if $P_a \cap Q_a$ is defined for each $a \in A$.

\item In other cases, $T \cap U$ is not defined.
\end{enumerate}
\end{definition}

The stable ordering $T \ssse U$ means that $U$ is more defined than $T$. On the other hand, the observational ordering $T \preceq U$ means that $U$ is more likely to converge than $T$ when they interact with other c-designs. The intersection $T \cap U$ corresponds to the common part of $T$ and $U$ delineated by $\Omega$ if $T \cap U$ is defined.
\begin{definition}[Incarnation]
  Let $\bd{T}$ be a behaviour. The incarnation $\val{U}_{\bd{T}}$ of $U$ in $\bd{T}$ is defined as $\val{U}_{\bd{T}} := \bigcap \{ U ' \in \bd{T} : U ' \ssse U \}$. We say $U$ is material in $\bd{T}$ if and only if $U = \val{U}_{\bd{T}}$ holds, and denote the set of all material designs in $\bd{T}$ by $\val{\bd{T}}$.
\end{definition}

Note that $\val{U}_{\bd{T}}$ is the minimal c-design $U'$ in $\bd{T}$ such that $U' \ssse U$ holds. Next, we define the \textit{shuffle} of two paths, which is a set of paths made by interleaving the actions in $p$ and $q$. Shuffling two paths is a key ingredient for the regularity of behaviour.
\begin{definition}[Shuffles]
  The shuffle of two paths and the shuffle of two sets of paths are sets of paths defined as follows:
  \begin{enumerate}
  \item Let $p$ and $q$ be paths. The shuffle $p \shuffle q$ of $p$ and $q$ is defined by distinguishing cases:
    \begin{itemize}
    \item If $p$ and $q$ are negative paths, then $p \shuffle q$ is the set of all paths $r$ such that any action in $r$ occurs in $p$ or $q$ and both of $r \uhr p = p$ and $r \uhr q = q$ hold,

    \item if $p$ and $q$ are positive paths of the same first action $\kappa^+$, that is, $p = \kappa^+p'$ and $q = \kappa^+q'$ hold, then $p \shuffle q$ is the set of all paths $r$ such that $r = \kappa^+ u$ holds for some $u \in p ' \shuffle q' $, and

    \item otherwise, $p \shuffle q$ is not defined.
    \end{itemize}
  \item Let $D$ and $D'$ be sets of paths. The shuffle $D \shuffle D'$ of $D$ and $D'$ is defined as the set of all paths $q$ such that for some $p \in D$ and $p' \in D'$ with $p \shuffle p'$ defined, $q \in p \shuffle p'$ holds.
  \end{enumerate}
\end{definition}

For instance, the shuffle of $x_1 | \ov{b} \langle y_1 , y_2 \rangle \;\;  a^{y_1} ( y_3 ) \;\; y_3 | \ov{c} \langle \rangle$ and $x_1 | \ov{b} \langle y_1 , y_2 \rangle \;\;  a^{y_2} ( y_4 ) \;\; y_4 | \ov{d} \langle \rangle$ is
\begin{center}
$\{ \Bigl( x_1 | \ov{b} \langle y_1 , y_2 \rangle \;\;  a^{y_1} ( y_3 ) \;\; y_3 | \ov{c} \langle \rangle \;\; a^{y_2} ( y_4 ) \;\; y_4 | \ov{d} \langle \rangle \Bigr) , \Bigl( x_1 | \ov{b} \langle y_1 , y_2 \rangle \;\; a^{y_2} ( y_4 ) \;\; y_4 | \ov{d} \langle \rangle \;\;  a^{y_1} ( y_3 ) \;\; y_3 | \ov{c} \langle \rangle \Bigr) \}$.
\end{center}
Note that $\epsilon \in D \shuffle D'$ holds if and only if $\epsilon$ belongs to both of $D$ and $D'$. We omit parentheses in consecutive application of $\shuffle$ because $p \shuffle q$ and $D \shuffle D'$ are associative.
\begin{definition}[Regular Behaviours]
  An a-behaviour $\bd{B}$ is regular if and only if the following conditions are satisfied: $(1)$ for any $T \in \val{ \bd{B} }$ and any positive-ended path $p$ of $T$, $p \in V ( \bd{B} )$ holds, $(2)$ for any $T \in \val{ \bd{B}^{\bot} }$ and any positive-ended path $p$ of $T$, $p \in V (\bd{B}^{\bot} )$ holds, and $(3)$ both $V ( \bd{B} )$ and $V ( \bd{B}^{\bot} )$ are closed under $\shuffle$.
\end{definition}

Roughly speaking, a behaviour $\bd{B}$ is regular if $\bd{B}$ is atomic and any positive-ended path of its material c-designs belongs to $V (\bd{B} )$, which is closed under $\shuffle$, and a similar condition holds for $\bd{B}^{\bot}$. When $\bd{N}$ is a negative a-behaviour, we denote by $V (x, \bd{N})$ the set of all paths obtained by replacing the address $x_0$ of the first actions of paths in $V ( \bd{N} )$ with $x$. Inspired by the results in \cite[\S~3.2]{pavaux2017c}, we define the second condition equivalent to the harmony condition as follows. This condition says that visitable paths made by a connective $\alpha$ from regular behaviours are dually decomposable as in our first condition.
\begin{definition}[Dual Decomposability of Visitable Paths]
  Let $\alpha$ be an $n$-ary connective. Visitable paths of $\alpha$ are dually decomposable if and only if for any negative regular behaviours $\bd{N}_1 ,\ldots, \bd{N}_n$ and any positive a-behaviours $\bd{P}_1 ,\ldots , \bd{P}_n$, we have
    \begin{itemize}
    \item $V ( \alpha^E \langle \bd{N}_1 ,\ldots, \bd{N}_n \rangle )  =   \{ \maltese \} \cup \bigcup_{a_i ( \vec{x}_i ) \in \alpha^I} x_0 | \ov{a_i} \langle \vec{x}_i\rangle ( V ( x_{(i,1)} , \bd{N}_{(i,1)} ) \shuffle \cdots \shuffle V ( x_{(i,k)} , \bd{N}_{(i,k)} ) )$,

    \item $V ( \alpha^I ( \bd{P}_1 ,\ldots, \bd{P}_n ) )  =  \{ \epsilon \} \cup \bigcup_{a_i ( \vec{x}_i )\in \alpha^E} a^{x_0}_i ( \vec{x}_i ) V ( [ (\bd{P}_{(i,1)})^{\bot} / x_{(i,1)},\ldots, (\bd{P}_{(i,k)})^{\bot} / x_{(i,k)} ]^{\bot} )$.
      
    \end{itemize}
\end{definition}

For any connective $\alpha = (\vec{z} , \alpha^I ,\alpha^E )$, we denote the connective $( \vec{z} , \beta^I , \beta^E )$ with $\beta^I = \alpha^E$ and $\beta^E = \alpha^I$ by $\alpha_{\bot}$. Intuitively, $\alpha_{\bot}$ is a connective whose introduction and elimination rules are harmonious with $\alpha^E$ and $\alpha^I$, respectively. When $p$ is a path of a c-design $T$, let $p^c$ be the c-design obtained by replacing with $\maltese$ any positive subdesign $P$ of $T$ such that $P = \Omega$ holds or the first action of elements of $\viewd{P}$ does not occur in $p$. Note that $p^c$ is a unique $\preceq$-maximal c-design $U$ such that $p$ is a path of $U$. We define the \textit{bi-view} $\langle s \rangle$ of an aj-sequence $s$ as $\langle \epsilon \rangle := \epsilon$, $\langle s \maltese \rangle = \langle s \rangle \maltese$ and
\[
\langle s \kappa \rangle :=
\begin{cases}
  \kappa , & \text{if $\kappa$ is initial in $s$ and $\kappa \neq \maltese$}, \\
  \langle s_0 \rangle \kappa , & \text{if $\kappa$ is justified by the last action of $s_0$ and $\kappa \neq \maltese$}.
\end{cases}
\]
The main lemma for our second characterisation of the harmony condition is as follows:
\begin{lemma}\label{mainlem}
  We have the following assertions:
  \begin{enumerate}
    
  \item If a connective $\alpha$ is dually decomposable, then $\alpha_{\bot}$ is dually decomposable.

  \item For any negative regular behaviours $\bd{N}_1 ,\ldots, \bd{N}_n$, we have
    \[
    V ( [ \bd{N}_{1} / x_{1},\ldots, \bd{N}_{n} / x_{n} ] ) =  V ( x_1 , \bd{N}_1 ) \shuffle \cdots \shuffle V ( x_n , \bd{N}_n ) .
    \]

  %\item Let $\bd{B}$ be an a-behaviour. If $p \in V (\bd{B})$ holds then for any positive-ended prefix $($resp. negative-ended prefix$)$ $s$ of $p$, we have $s \in V (\bd{B})$ $($resp. $s \maltese \in V (\bd{B}) )$.
    
  \end{enumerate}
\end{lemma}
\begin{proof}
  (1.) Assume that $\alpha$ is dually decomposable. We have
  \begin{eqnarray*}
    ( \alpha_{\bot} )^I ( \bd{P}_1 ,\ldots , \bd{P}_n )  & = & \bigcap_{ a_i (\vec{x}_i ) \in \alpha^E } ( \ov{a_i} \langle \bd{P}_1^{\bot} ,\ldots , \bd{P}_n^{\bot} \rangle^{\bot} )  =  ( \bigcup_{ a_i (\vec{x}_i ) \in \alpha^E } \ov{a_i} \langle \bd{P}_1^{\bot} ,\ldots , \bd{P}_n^{\bot} \rangle )^{\bot} \\
    & = & \alpha^E \langle \bd{P}_1^{\bot} ,\ldots , \bd{P}_n^{\bot} \rangle^{\bot}  =  ( \alpha^I (\bd{P}_1 ,\ldots , \bd{P}_n )^{\sC} \cup \{ \maltese \} )^{\bot} \\
    & = & ( \alpha^I (\bd{P}_1 ,\ldots , \bd{P}_n )^{\sC} )^{\bot} \cap \{ \maltese \}^{\bot}  =  ( \alpha^I (\bd{P}_1 ,\ldots , \bd{P}_n )^{\sC} )^{\bot} \\
    & = & \bigcap_{ a_i (\vec{x}_i ) \in \alpha^I } ( \ov{a_i} \langle \bd{P}_1^{\bot} ,\ldots , \bd{P}_n^{\bot} \rangle^{\bot} ) = \bigcap_{ a_i (\vec{x}_i ) \in ( \alpha_{\bot} )^E } ( \ov{a_i} \langle \bd{P}_1^{\bot} ,\ldots , \bd{P}_n^{\bot} \rangle^{\bot} ) ,
  \end{eqnarray*}
  so the one half of the dual decomposability of $\alpha_{\bot}$ holds. On the other hand, we have
  \begin{eqnarray*}
    ( \alpha_{\bot} )^E \langle \bd{N}_1 ,\ldots , \bd{N}_n \rangle & = & ( \bigcup_{ a_i ( \vec{x}_i  ) \in \alpha^I } \ov{a_i} \langle \bd{N}_1 ,\ldots , \bd{N}_n \rangle )^{\bot\bot} = ( \bigcap_{ a_i ( \vec{x}_i  ) \in \alpha^I } ( \ov{a_i} \langle \bd{N}_1 ,\ldots , \bd{N}_n \rangle^{\bot} )  )^{\bot} \\
    & = & \alpha^I ( \bd{N}_1^{\bot} ,\ldots , \bd{N}_n^{\bot}  )^{\bot}.
  \end{eqnarray*}
  It is obvious that $( \bigcup_{ a_i ( \vec{x}_i ) \in \alpha^E } \ov{a_i} \langle \bd{N}_1 ,\ldots , \bd{N}_n \rangle ) \cup \{ \maltese \} \subseteq \alpha^I ( \bd{N}_1^{\bot} ,\ldots , \bd{N}_n^{\bot}  )^{\bot}$ holds, because $\alpha$ is dually decomposable. Conversely, assume that $P \in \alpha^I ( \bd{N}_1^{\bot} ,\ldots , \bd{N}_n^{\bot}  )^{\bot}$ holds. If $P = \maltese$ holds then the assertion is obvious, so let $P = x_0 | \ov{b} \langle M_1 , \ldots , M_m \rangle$ be the case. By the dual decomposability of $\alpha$ again, we have $\sum_{a_i ( \vec{x}_i ) \in \alpha^E}. \maltese \in \alpha^I ( \bd{N}_1^{\bot} ,\ldots , \bd{N}_n^{\bot}  )$. There is a negative action $a_j ( \vec{x}_j ) \in \alpha^E$ with $P = x_0 | \ov{a_j} \langle M_{(j , 1)} , \ldots , M_{(j , k)} \rangle$ because $P \bot \sum_{a_i ( \vec{x}_i ) \in \alpha^E}. \maltese$ holds. Moreover, for any $Q \in \bd{N}_{(j,l)}^{\bot}$, $N := a_j ( \vec{x}_j ) . Q [x_{(j,l)} / x_0 ] + b_1 ( \vec{x}_{b_1} ). \maltese + \cdots + b_{m} ( \vec{x}_{b_m} ).\maltese$ belongs to $\alpha^I ( \bd{N}_1^{\bot} ,\ldots , \bd{N}_n^{\bot}  )$, where $\alpha^E \setminus \{ a_j ( \vec{x}_j ) \} = \{ b_1 ( \vec{x}_{b_1} ),\ldots , b_m ( \vec{x}_{b_m} )\}$ holds. We have $P [N / x_0] \rwt Q[M_{(j , l)} / x_{(j , l)}] \rwt \maltese$ by $P \bot N$, hence $M_{(j , l)} \in \bd{N}_{(j,l)}$ holds for any $l$ with $1 \leq l \leq k$. Therefore, $P \in \bigcup_{ a_i ( \vec{x}_i ) \in \alpha^E } \ov{a_i} \langle \bd{N}_1 ,\ldots , \bd{N}_n \rangle$ holds. Now we have $( \bigcup_{ a_i ( \vec{x}_i ) \in \alpha^E } \ov{a_i} \langle \bd{N}_1 ,\ldots , \bd{N}_n \rangle ) \cup \{ \maltese \} = \alpha^I ( \bd{N}_1^{\bot} ,\ldots , \bd{N}_n^{\bot}  )^{\bot}$, so it follows that $( \alpha_{\bot} )^E \langle \bd{N}_1 ,\ldots , \bd{N}_n \rangle = ( \alpha_{\bot} )^I ( \bd{N}_1^{\bot} ,\ldots , \bd{N}_n^{\bot} )^{\sC} \cup \{ \maltese \}$ holds.

  (2.) Put $\bd{G} := [\bd{N}_1 / x_1 ,\ldots , \bd{N}_n / x_n]$.

  ($\subseteq$) Let $p \in V ( \bd{G} )$ be the case. If $p = \epsilon$ holds then the assertion is obvious, so let $p$ be non-empty. By definition, $p = \iseq{[G]}{P}$ holds for some $[G] = [N_1 / x_1 ,\ldots , N_n / x_n] \in \bd{G}$ and some c-design $P \in \bd{G}^{\bot}$. Fix an arbitrary $i$ with $1 \leq i \leq n$, and put $[G '] := [G] \setminus [N_i / x_i]$. By the definition of $\fk{Cut}$, we have $P[G '] \in \Cut{[G ']}{P}$ whether $\fv{P} = \{ x_1 , \ldots , x_n \}$ holds or $\fv{P} \subset \{ x_1 , \ldots , x_n \}$ holds. Then, by Proposition \ref{asspaths}, we have
  \[
  p \uhr [ N_i /x_i ] = \iseq{[G]}{P} \uhr [ N_i / x_i ] = \iseq{[N_i / x_i]}{ \otv{ \Cut{[G ']}{P} } } = \iseq{ [N_i / x_i] }{ \otv{ P[G'] } }.
  \]
  On the other hand, we have $\otv{ \otv{ P[G '] }[M / x_ i] } = \otv{ \otv{ P[G'] } [ \otv{M} / x_i ] } = \otv{ P[ G '][ M / x_i ] }= \maltese$ for any $M \in \bd{N}_i$ by Theorem \ref{ass}, hence $\otv{ P[G'] } \in \bd{N}_i^{\bot}$ holds. Moreover, the address of the first action of $p \uhr [ N_i /x_i ]$ is $x_i$ because $\iseq{[N_i / x_i]}{ \otv{ P[G'] } }$ is a path of $[N_i / x_i]$ by Lemma \ref{formainone}.(1). Therefore, $p \uhr [ N_i /x_i ] = \iseq{[N_i / x_i]}{ \otv{P[G'] } } \in V ( x_i , \bd{N}_i )$ holds for any $i$, so $p$ belongs to $V ( x_1 , \bd{N}_1 ) \shuffle \cdots \shuffle V ( x_n , \bd{N}_n )$.

  ($\supseteq$) Let $p \in V ( x_1 , \bd{N}_1 ) \shuffle \cdots \shuffle V ( x_n , \bd{N}_n )$ be the case, and assume that we have shown $\dual{p}^{\: c} \in \bd{G}^{\bot}$. By $p \in V ( x_1 , \bd{N}_1 ) \shuffle \cdots \shuffle V ( x_n , \bd{N}_n )$, $p$ is a path of $[ \vec{N} ]$ for some $[ \vec{N} ] \in \bd{G}$. Then, $p \in V ( \bd{G} )$ holds because one can show $p = \iseq{ [ \vec{N} ] }{ \dual{p}^{\: c} }$. Therefore, it suffices to verify that $\dual{p}^{\: c} \in \bd{G}^{\bot}$ holds.

  We suppose that $\dual{p}^{\: c} \not\in \bd{G}^{\bot}$ holds and deduce a contradiction. By $\dual{p}^{\: c} \not\in \bd{G}^{\bot}$, $[ \vec{N} ] \bot \dual{p}^{\: c}$ does not hold for some $[ \vec{N} ] = [ N_1 / x_1 ,\ldots , N_n / x_n ] \in \bd{G}$. The interaction of $[ \vec{N} ]$ and $\dual{p}^{\: c}$ cannot be infinite because $\dual{p}^{\: c}$ is a completion by means of $\maltese$ and $p$ is a shuffle of some visitable paths.
  %Moreover, if $\ov{ s \kappa^- }$ is a path of $[ \vec{N} ]$ and $s$ is a path of $\viewpf{ p_1 }^c$, then $s \kappa^-$ is a path of $\viewpf{ p_1 }^c$.
  By Lemma \ref{formainone}.(2), there are a path $t$ of $[\vec{N}]$ and a negative action $\kappa^-$ such that $\ov{ t \kappa^- }$ is a path of $\dual{p}^{\: c}$ and $t \kappa^-$ is not a path of $[ \vec{N} ]$. Therefore, there is a path $t$ satisfying the following property $( \ast )$: for some anti-design $[\vec{N}] \in \bd{G}$, (i) $[\vec{N}] \bot \dual{p}^{\: c}$ does not hold, (ii) $t$ is a path of $[\vec{N}]$ and (iii) for some negative action $\kappa^-$, $\ov{ t \kappa^- }$ is a path of $\dual{p}^{\: c}$ and $t \kappa^-$ is not a path of $[ \vec{N} ]$. Choose a minimal path $t$ with respect to length such that $t$ satisfies the property $( \ast )$. We show the following claims (a)--(d), and the claim (d) contradicts the property $( \ast )$ of $t$.
  \begin{enumerate}
  \item[(a)] $\dual{t}^{\; c} \in \bd{G}^{\bot}$ holds,

  \item[(b)] $t \in V ( x_1 , \bd{N}_1 ) \shuffle \cdots \shuffle V ( x_n , \bd{N}_n )$ holds,

  \item[(c)] for any $v \in V ( x_1 , \bd{N}_1 ) \shuffle \cdots \shuffle V ( x_n , \bd{N}_n )$ and any $\kappa^-_1$ such that $\ov{ v \kappa^-_1 }$ is a path of $\dual{p}^{\: c}$, we have $v \kappa^-_1 \maltese \in V ( x_1 , \bd{N}_1 ) \shuffle \cdots \shuffle V ( x_n , \bd{N}_n )$,

  \item[(d)] $t\kappa^-$ is a path of $[ \vec{N} ]$.
  \end{enumerate}

  (a) Suppose that $\dual{t}^{\; c} \not\in \bd{G}^{\bot}$ holds, hence $[ \vec{N}_1 ] \bot \dual{t}^{\; c}$ does not hold for some $[ \vec{N}_1 ] \in \bd{G}$. By Lemma \ref{formainone}.(2), there are a path $t_0$ and a negative action $\kappa^-_2$ such that $t_0$ is a path of $[ \vec{N}_1 ]$, $\ov{ t_0 \kappa^-_2 }$ is a path of $\dual{t}^{\; c}$ and $t_0 \kappa^-_2$ is not a path of $[ \vec{N}_1 ]$. One can see that $\ov{ t_0 \kappa^-_2 }$ is a path of $\dual{p}^{\: c}$, because views of $\dual{t}^{\; c}$ are views of $\dual{p}^{\: c}$. Then, $[ \vec{N}_1 ] \bot \dual{p}^{\: c}$ does not hold, otherwise $t_0 \kappa^-_2$ would be a prefix of $\iseq{[ \vec{N}_1 ]}{\dual{p}^{\: c}}$ and so $t_0 \kappa^-_2$ is a path of $[ \vec{N}_1 ]$ by Lemma \ref{formainone}.(1). Moreover, $t_0$ is strictly shorter than $t$ because $\ov{ t_0 \kappa^-_2 }$ is a path of $\dual{t}^{\; c}$. This contradicts the minimality of $t$ with respect to length.

  (b) We have $t = \iseq{ [ \vec{N} ] }{ \dual{t}^{\; c} }$ because $t$ is a path of $[ \vec{N} ]$ and $\dual{t}$ is a path of $\dual{t}^{\; c}$. Then, as in the proof of ($\subseteq$) above, we have $t \uhr [ N_i / x_i ] \in V ( x_i , \bd{N}_i )$ for any $i$ by the claim (a).

  (c) Let $v$ be an element of $V ( x_1 , \bd{N}_1 ) \shuffle \cdots \shuffle V ( x_n , \bd{N}_n )$, and $\kappa^-_1$ be a negative action such that $\ov{ v \kappa^-_1 }$ is a path of $\dual{p}^{\: c}$. By $p ,v \in V ( x_1 , \bd{N}_1 ) \shuffle \cdots \shuffle V ( x_n , \bd{N}_n )$, there are $p_1 , \ldots , p_n$ and $v_1 , \ldots , v_n$ such that $p \in p_1 \shuffle \cdots \shuffle p_n$ and $v \in v_1 \shuffle \cdots \shuffle v_n$ hold and we have $p_i , v_i \in V ( x_i , \bd{N}_i )$ for any $i$ with $1 \leq i \leq n$. Assume that $\kappa^-_1$ is an action in $p_j$. It suffices to show $v_j \kappa^-_1 \maltese \in V ( x , \bd{N}_j )$. By the definition of bi-views, we have $\langle v_j \kappa^-_1 \rangle = \langle v \kappa^-_1 \rangle$. Moreover, for the prefix $p' \kappa^-_1$ of $p$, $\langle v \kappa^-_1 \rangle = \langle p ' \kappa^-_1 \rangle$ holds because $\ov{ v \kappa^-_1 }$ is a path of $\dual{p}^{\: c}$. We have $\langle p ' \kappa^-_1 \rangle = \langle p_j ' \kappa^-_1 \rangle$ for the prefix $p_j ' \kappa^-_1$ of $p_j$ by the definition of bi-views again, hence $\langle v_j \kappa^-_1  \rangle = \langle p_j ' \kappa^-_1 \rangle$ holds. We have $p_j ' \kappa^-_1 \maltese \in V ( x_j , \bd{N}_j )$ by $p_j \in V ( x_j , \bd{N}_j )$ and Lemma \ref{formainone}.(3). Then, $\langle v_j \kappa^-_1 \maltese \rangle$ is a path of some $N_j  \in  \val{ \bd{N}_j }$ because $\langle v_j \kappa^-_1  \rangle = \langle p_j ' \kappa^-_1 \rangle$ holds, hence we have $\langle v_j \kappa^-_1 \maltese \rangle \in V ( x_j , \bd{N}_j )$ by the regularity of $\bd{N}_j$. Then, $v_j \kappa^-_1 \maltese \in v_j \shuffle \langle v_j \kappa^-_1 \maltese \rangle$ holds, so we have $v_j \kappa^-_1 \maltese \in V ( x_j , \bd{N}_j )$ by the closedness of $V ( x_j , \bd{N}_j )$ under $\shuffle$.

  (d) By the claim (b), we have $t \in V ( x_1 , \bd{N}_1 ) \shuffle \cdots \shuffle V ( x_n , \bd{N}_n )$. Moreover, $\ov{ t \kappa^- }$ is a path of $\dual{p}^{\: c}$ by the property $( \ast )$ of $t$, hence $t \kappa^- \maltese \in V ( x_1 , \bd{N}_1 ) \shuffle \cdots \shuffle V ( x_n , \bd{N}_n )$ holds by the claim (c). Then, we have $t_i \kappa^- \maltese \in V ( x_i , \bd{N}_i )$ for some $\bd{N}_i$ and some $t_i$. The sequence $t_i$ is a path of $N_i \in \bd{N}_i$, so $t_i \kappa^-$ is a prefix of a path of $N_i$. Therefore, $\view{ t\kappa^- }$ is a view of $N_i$ because we have $\view{ t\kappa^- } = \view{ t_i\kappa^- }$. The sequence $t$ is a path of $[ \vec{N} ]$, hence $t\kappa^-$ is a path of $[ \vec{N} ]$.
\end{proof}

By the lemma above, we have the following proposition:
\begin{proposition}\label{protwo}
  $(1)$ If a connective $\alpha$ is dually decomposable, then visitable paths of $\alpha$ are dually decomposable. $(2)$ if visitable paths of a connective $\alpha$ are dually decomposable, then $\alpha$ satisfies the harmony condition.
\end{proposition}
\begin{proof}
  (1.) First, we show
  \[
  V ( \alpha^E \langle \bd{N}_1 ,\ldots, \bd{N}_n \rangle )  =  \{ \maltese \} \cup \bigcup_{a_i ( \vec{x}_i ) \in \alpha^I} x_0 | \ov{a_i} \langle \vec{x}_i\rangle (  V ( x_{(i,1)} , \bd{N}_{(i,1)} ) \shuffle \cdots \shuffle V ( x_{(i,k)} , \bd{N}_{(i,k)} ) ).
  \]
  
  ($\subseteq$) Assume that $p \in V ( \alpha^E \langle \bd{N}_1 ,\ldots, \bd{N}_n \rangle )$ holds. If $p = \maltese$ holds then the assertion is trivial and so let $p \neq \maltese$ be the case. By the dual decomposability of $\alpha$, $p = \iseq{P}{N}$ holds for some $a_i (\vec{x}_i ) \in \alpha^I$, some $P = x_0 | \ov{a_i} \langle N_{(i,1)} , \ldots , N_{(i,k)} \rangle$ with $N_{(i,j)} \in \bd{N}_{(i,j)}$ for any $j$, and some $N = \sum b ( \vec{y}_b ).P_b \in \alpha^E \langle \bd{N}_1 ,\ldots, \bd{N}_n \rangle^{\bot}$. Therefore, by the definition of interaction sequences and the renaming of bound variables if necessary, $p$ is equal to $x_0 | \ov{a_i} \langle \vec{x}_{i}\rangle p'$ with $p' \in V ( [ \bd{N}_{(i,1)} / x_{(i,1)},\ldots, \bd{N}_{(i,k)} / x_{(i,k)} ] )$. By Lemma \ref{mainlem}.(2), we have
  \[
  p \in x_0 | \ov{a_i} \langle \vec{x}_i\rangle (  V ( x_{(i,1)} , \bd{N}_{(i,1)} ) \shuffle \cdots \shuffle V ( x_{(i,k)} , \bd{N}_{(i,k)} ) ).
  \]
  
  ($\supseteq$) The case of $\maltese$ is obvious.
  %Assume that $p \in x_0 | \ov{a_i} \langle \vec{x}_i\rangle ( V ( x_{(i,1)} , \bd{N}_{(i,1)} ) \shuffle \cdots \shuffle V ( x_{(i,k)} , \bd{N}_{(i,k)} ) )$ holds for some $a_i ( \vec{x}_i ) \in \alpha^I$. If $p = x_0 | \ov{a_i} \langle \vec{x}_i \rangle$ holds then we have $p = \iseq{x_0 | \ov{a_i} \langle N_1 , \ldots , N_k \rangle }{\sum b ( \vec{y}_b ). \maltese}$ for some arbitrary $N_1 \in \bd{N}_{(i,1)} , \ldots , N_k \in \bd{N}_{(i,k)}$. Note that any negative behaviour $\bd{M}$ is non-empty because $\sum b ( \vec{y}_b ). \maltese \in \bd{M}^{\bot\bot}$ holds. Since we have $x_0 | \ov{a_i} \langle N_1 , \ldots , N_k \rangle \in \alpha^E \langle \bd{N}_1 , \ldots , \bd{N}_n \rangle$ by the dual decomposability of $\alpha$ and $\sum b ( \vec{y}_b ). \maltese \in \alpha^E \langle \bd{N}_1 , \ldots , \bd{N}_n \rangle^{\bot}$ holds, we have $p \in V ( \alpha^E \langle \bd{N}_1 , \ldots , \bd{N}_n \rangle )$.
  Assume that $p \in x_0 | \ov{a_i} \langle \vec{x}_i\rangle  ( V ( x_{i(1)} , \bd{N}_{i(1)} ) \shuffle \cdots \shuffle V ( x_{i(k)} , \bd{N}_{i(k)} )$ holds for some $a_i ( \vec{x}_i ) \in \alpha^I$. Then, we have $p = x_0 | \ov{a_i} \langle \vec{x}_i\rangle \iseq{[G]}{P}$ for some
  \[
  [G] = [N_1 / x_{(i,1)} , \ldots , N_k / x_{(i,k)} ] \in \bd{G} := [ \bd{N}_{(i,1)} / x_{(i,1)},\ldots, \bd{N}_{(i,k)} / x_{(i,k)} ]
  \]
  and some $P \in \bd{G}^{\bot}$ by Lemma \ref{mainlem}.(2). Define $Q := x_0 | \ov{a_i} \langle N_1 ,\ldots , N_k \rangle$, then $Q \in \alpha^I ( \bd{N}_1^{\bot} ,\ldots , \bd{N}_n^{\bot} )^{\sC}$ holds and so we have $Q \in \alpha^E \langle \bd{N}_1 ,\ldots , \bd{N}_n \rangle$ by the dual decomposability of $\alpha$. On the other hand, define $M := a_i ( \vec{x}_i ).P + \sum_{\beta}.\maltese$ with $\beta = \alpha^I \setminus \{ a_i ( \vec{x}_i ) \}$, then we have $M \in \alpha^E \langle \bd{N}_1 ,\ldots , \bd{N}_n \rangle^{\bot}$ by the dual decomposability again. By $p = \iseq{Q}{M}$, the assertion holds.

  Next, we show
  \[
  V ( \alpha^I ( \bd{P}_1 ,\ldots, \bd{P}_n ) )  =  \{ \epsilon \} \cup \bigcup_{a_i ( \vec{x}_i )\in \alpha^E} a^{x_0}_i ( \vec{x}_i ) V ( [ (\bd{P}_{(i,1)})^{\bot} / x_{(i,1)},\ldots, (\bd{P}_{(i,k)})^{\bot} / x_{(i,k)} ]^{\bot} ).
  \]

  ($\subseteq$) Assume that $p \in V ( \alpha^I ( \bd{P}_1 ,\ldots, \bd{P}_n ) )$ holds. If $p = \epsilon$ holds then the assertion obviously holds, so let $p$ be non-empty. By definition, we have $p = \iseq{N}{P}$ with $N \in \alpha^I ( \bd{P}_1 ,\ldots, \bd{P}_n )$ and $P \in \alpha^I ( \bd{P}_1 ,\ldots, \bd{P}_n )^{\bot}$. By the dual decomposability of $\alpha$, $N = \sum a ( \vec{y}_a ) . P_a$ holds and $P_{a_i} \in [\bd{P}_{(i,1)}^{\bot} / x_{(i,1)} ,\ldots ,\bd{P}_{(i,k)}^{\bot} / x_{(i,k)} ]^{\bot}$ holds for any $a_i (\vec{x}_i ) \in \alpha^E$. We have the following equation $( \ast )$
  \[
  \alpha^I ( \bd{P}_1 ,\ldots, \bd{P}_n )^{\bot}  =  ( \bigcup_{ a_i ( \vec{x}_i ) \in \alpha^I } \ov{a_i} \langle \bd{P}_{(i,1)}^{\bot} , \ldots , \bd{P}_{(i,k)}^{\bot} \rangle )^{\bot\bot}  =  (\alpha_{\bot})^E \langle \bd{P}_1^{\bot} ,\ldots, \bd{P}_n^{\bot} \rangle  =  (\alpha_{\bot} )^I ( \bd{P}_1 ,\ldots, \bd{P}_n )^{\sC} \cup \{ \maltese \}
  \]
  by the dual decomposability of $\alpha$ and Lemma \ref{mainlem}.(1), hence $P = x_0 | \ov{a_i} \langle N_1 ,\ldots , N_k \rangle$ holds for some $a_i ( \vec{x}_i ) \in \alpha^E$ and $N_j \in \bd{P}_{(i,j)}^{\bot}$ holds for any $j$. Therefore, $p \in a^{x_0}_i ( \vec{x}_i ) V ( [ (\bd{P}_{(i,1)})^{\bot} / x_{(i,1)},\ldots, (\bd{P}_{(i,k)})^{\bot} / x_{(i,k)} ]^{\bot} )$ holds by the definition of interaction sequences.

  ($\supseteq$) It suffices to consider the case of non-empty sequences. Assume that
  \[
  p \in \bigcup_{a_i ( \vec{x}_i  ) \in \alpha^E } a^{x_0}_i ( \vec{x}_i ) V ( [ (\bd{P}_{(i,1)})^{\bot} / x_{(i,1)},\ldots, (\bd{P}_{(i,k)})^{\bot} / x_{(i,k)} ]^{\bot} )
  \]
  holds and put $\bd{G} := [ (\bd{P}_{(i,1)})^{\bot} / x_{(i,1)},\ldots, (\bd{P}_{(i,k)})^{\bot} / x_{(i,k)} ]$. By definition, there are $P \in \bd{G}^{\bot}$ and
  \[
  [H] = [N_1 / x_{(i , 1)} ,\ldots , N_k / x_{(i , k)} ] \in \bd{G}
  \]
  such that $p = a^{x_0}_i ( \vec{x}_i ) \iseq{P}{[H]}$ and $a_i (\vec{x}_i ) \in \alpha^E$ hold. Define $N := a_i ( \vec{x}_i ) . P + \sum_{\beta} . \maltese$ with $\beta = \alpha^E \setminus \{ a_i (\vec{x}_i ) \}$ and $Q := x_0 | \ov{a_i} \langle N_1 ,\ldots , N_k \rangle$. By the dual decomposability of $\alpha$, we have $N \in \alpha^I ( \bd{P}_1 ,\ldots, \bd{P}_n )$. Moreover, by $Q \in (\alpha_{\bot} )^I ( \bd{P}_1 ,\ldots, \bd{P}_n )^{\sC}$, we have $Q \in \alpha^I ( \bd{P}_1 ,\ldots, \bd{P}_n )^{\bot}$ by the equation $( \ast )$ in the previous case. Therefore, we have $p \in V ( \alpha^I ( \bd{P}_1 ,\ldots, \bd{P}_n ) )$ because $p = \iseq{N}{Q}$ holds.

  (2.) Assume that $\alpha$ does not satisfy the harmony condition. We suppose that visitable paths of $\alpha$ is dually decomposable, and deduce a contradiction. If there is a negative action $a_i (x_i ) \in \alpha^I \setminus \alpha^E$, consider a path
  \[
  p = x_0 | \ov{a_i} \langle \vec{x}_i \rangle p' \in \bigcup_{a_j ( \vec{x}_j ) \in \alpha^I} x_0 | \ov{a_j} \langle \vec{x}_j\rangle ( V ( x_{(j,1)} , \bd{N}_{(j,1)} ) \shuffle \cdots \shuffle V ( x_{(j,k)} , \bd{N}_{(j,k)} ) ).
  \]
  We have $p \in V ( \alpha^E \langle \bd{N}_1 ,\ldots , \bd{N}_n \rangle )$, so $p = \iseq{ Q }{ M }$ holds for some $Q \in \alpha^E \langle  \bd{N}_1 ,\ldots , \bd{N}_n \rangle$. The c-design $Q$ is of the form $x_0 | \ov{a_i} \langle N_{(i,1)} , \ldots N_{(i,m)} \rangle$ with $a_i (x_i ) \in \alpha^I \setminus \alpha^E$ because $p$ is a path of $Q$, hence one can find
  \[
  N \in \alpha^E \langle \bd{N}_1 ,\ldots , \bd{N}_n \rangle^{\bot} = ( \bigcup_{ a_j (\vec{x}_j ) \in \alpha^E } \ov{a_j } \langle \bd{N}_{(j,1)} ,\ldots , \bd{N}_{(j,k)} \rangle )^{\bot}
  \]
  such that $Q \bot N$ does not hold. Contradiction.

  If there is a negative action $a_i (x_i ) \in \alpha^E \setminus \alpha^I$, consider a path
  \[
  p = a_i^{x_0} ( \vec{x}_i ) p' \in \bigcup_{a_j ( \vec{x}_j )\in \alpha^E} a^{x_0}_j ( \vec{x}_j ) V ( [ (\bd{P}_{(j,1)})^{\bot} / x_{(j,1)},\ldots, (\bd{P}_{(j,k)})^{\bot} / x_{(j,k)} ]^{\bot} ) .
  \]
  We have $p \in V ( \alpha^I ( \bd{P}_1 ,\ldots , \bd{P}_n ) )$, so $p = \iseq{ M }{ Q }$ holds for some $Q \in \alpha^I (  \bd{P}_1 ,\ldots , \bd{P}_n )^{\bot}$. The dual $\dual{p} = \iseq{Q}{M}$ is a path of $Q$, hence $Q$ is of the form $x_0 | \ov{a_i} \langle N_{(i,1)} , \ldots N_{(i,m)} \rangle$ with $a_i (x_i ) \in \alpha^E \setminus \alpha^I$. Therefore, one can find
  \[
  N \in \alpha^I (  \bd{P}_1 ,\ldots , \bd{P}_n )^{\bot\bot} =  \bigcap_{ a_j (\vec{x}_j ) \in \alpha^I }  \ov{a_j } \langle \bd{P}_{(j,1)}^{\bot} ,\ldots , \bd{P}_{(j,k)}^{\bot} \rangle^{\bot}
  \]
  such that $Q \bot N$ does not hold. Contradiction.
\end{proof}

Our two characterisations of the harmony condition are obtained by Propositions \ref{proone} and \ref{protwo}.
\begin{corollary}[Characterisation of Harmony]\label{main}
  Let $\alpha$ be a connective. The following three assertions are equivalent: $(1)$ $\alpha$ satisfies the harmony condition, $(2)$ $\alpha$ is dually decomposable and $(3)$ visitable paths of $\alpha$ is dually decomposable.
\end{corollary}

\section{Concluding Remarks and Future Work}
By means of Computational Ludics, we have first reformulated the inversion principle and the recovery principle into the harmony condition. Then, we have shown that the harmony condition is equivalent to both the dual decomposability of connectives and the dual decomposability of visitable paths.

However, a thorough analysis of the fundamental features of proof-theoretic semantics by means of the Computational Ludics tools is far from being definitely achieved. First, the proof-theoretic semantics literature has considered other principles such as \textit{deducibility of identicals} or the \textit{uniqueness} (see \cite{NP2015}) to capture the necessary condition that a set of rules has to satisfy to define a meaningful and logical connective. Examining how these principles can be reformulated in Computational Ludics would be a crucial step for future works. Second, as shown in \cite{BT2010}, in Computational Ludics it is possible to have a logical connective (i.e. a connective satisfying the harmony condition) for the non-linear case, which does not enjoy the internal completeness. To fully appreciate the relationship between the logicality and the internal completeness, we will explore the non-linear case.

Concerning the philosophical scope of our work, let us remark that, as we mentioned in the introduction, proof-theoretic semantics has been traditionally developed within the framework of natural deduction. However, as noted in \cite[\S~1.2]{SchroederHeister18}, natural deduction is somehow ``biased towards intuitionistic logic''. The possibility of associating each connective to a set of introduction rules and then justifying a corresponding set of elimination rules by means of detour reduction works straightforwardly when the intuitionistic rules are considered (on the contrary, the classical rule of \textit{reductio ad absurdum}, or of indirect proof, cannot be easily classified as an introduction rule nor as an elimination rules, and this makes it difficult  to define a suitable notion of detour for it; see \cite{GuerrieriNaibo20}). A monistic point of view is thus often associated with proof-theoretic semantics, according to which intuitionistic logic is the only right and meaningful logic.
%But, as remarked by \cite[\S~1.2]{SchroederHeister18}, natural deduction ``is biased towards intuitionistic logic'', since the possibility of associating each connective to a set of introduction rules and a corresponding set of elimination rules -- together with the possibility of defining a notion of detour with respect to these sets of introduction and elimination rules -- works particularly well in the case of intuitionistic logic (the classical rule of reductio ad absurdum, or of indirect proof, destroy such a symmetry between introduction and elimination rules, and it makes hard to define a suitable notion of detour; see ).
The analysis of harmony that we have offered here aims to show that when the notion of proof is formalised within a framework different from natural deduction, then other connectives---different from the intuitionistic ones---can be justified. We took here Computational Ludics as an alternative framework to natural deduction, and we showed that this choice allows for the justification of linear connectives. In this sense, our work can eventually be seen as a contribution to the idea that proof-theoretic semantics is compatible with a \textit{pluralistic} rather than a monistic view of logic. We also claim that our use of Computational Ludics as an alternative framework to natural deduction is legitimised by the fact that it allows us to obtain a more perspicuous formulation of harmony than the one that is usually proposed in the case of (intuitionistic) natural deduction.

\nocite{*}
\bibliographystyle{eptcs}
\bibliography{harmonyludics}

\end{document}